\documentclass[11pt,a4paper]{article}
\usepackage{amsmath,amssymb,amsthm,amsfonts,latexsym,bm,bbm,xspace,graphicx,float,mathtools,mathdots,physics}
\usepackage{braket,caption,subcaption,ellipsis,textcomp,hhline,pifont,combelow,booktabs}
\usepackage{wrapfig}

\usepackage[linesnumbered,boxed, vlined]{algorithm2e}
    \DontPrintSemicolon
  \SetKwInOut{Input}{Input}
  \SetKwInOut{Output}{Output}

\usepackage[margin=1in,letterpaper]{geometry}
\usepackage{thmtools}
\usepackage{thm-restate}

\usepackage[dvipsnames]{xcolor}
\definecolor{ForestGreen}{rgb}{0.1333,0.5451,0.1333}
\definecolor{DarkRed}{rgb}{0.65,0,0}
\definecolor{Red}{rgb}{1,0,0}
\definecolor{DarkRed}{rgb}{0.5,0.1,0.1}
\definecolor{DarkBlue}{rgb}{0.1,0.1,0.5}
\usepackage[linktocpage=true,
pagebackref=true,colorlinks,
linkcolor=NavyBlue,citecolor=ForestGreen,
bookmarks,bookmarksopen,bookmarksnumbered]{hyperref}

\declaretheorem[numberwithin=section]{lemma}

\declaretheorem[numberlike=lemma]{fact}
\declaretheorem[numberlike=lemma]{proposition}

\theoremstyle{definition}
\declaretheorem[numberlike=lemma]{definition}

\usepackage[noabbrev,capitalise,nameinlink]{cleveref}
\crefname{theorem}{Theorem}{Theorems}
\crefname{section}{Section}{Sections}
\crefname{lemma}{Lemma}{Lemmas}
\crefname{observation}{Observation}{Observations}
\crefname{algorithm}{Algorithm}{Algorithms}
\crefname{step}{Step}{Steps}
\crefname{fact}{Fact}{Facts}
\crefname{claim}{Claim}{Claims}
\crefname{part}{Property}{Properties}
\crefname{issue}{Challenge}{Challenges}
\crefname{tech-issue}{Technical Challenge}{Technical Challenges}

\newcommand{\owncref}[2]{\hyperref[#2]{#1~\ref*{#2}}}
\newcommand{\owntocref}[3]{\hyperref[#2]{#1~\ref*{#2}} to~\ref{#3}}

\usepackage[textsize=scriptsize,backgroundcolor=gray!7,bordercolor=gray,linecolor=gray]{todonotes}
\iffalse

\newcommand{\maxime}[2][]{\todo[#1]{MF: #2}}
\newcommand{\alex}[2][]{\todo[#1]{AN: #2}}
\newcommand{\magnus}[2][]{\todo[#1]{MH: #2}}

\else

\newcommand{\maxime}[2][]{}
\newcommand{\alex}[2][]{}
\newcommand{\magnus}[2][]{}

\fi

\renewcommand{\paragraph}[1]{\medskip\noindent{\bf #1}\xspace}
\newcommand{\Qed}[1]{\ensuremath{\qedsymbol_{\,\,\textnormal{\cref{#1}}}}}

\renewcommand{\epsilon}{\ensuremath{\varepsilon}}
\newcommand{\eps}{\ensuremath{\varepsilon}}

\DeclarePairedDelimiter{\card}{\lvert}{\rvert}%%
\let\set\relax
\DeclarePairedDelimiter{\set}{\lbrace}{\rbrace}%%
\DeclarePairedDelimiter{\parens}{\lparen}{\rparen}%%
\DeclarePairedDelimiter{\floor}{\lfloor}{\rfloor}%%
\DeclarePairedDelimiter{\ceil}{\lceil}{\rceil}%%

\let\Pr\relax
\DeclareMathOperator*{\Pr}{\ensuremath{\mathsf{Pr}}}
\DeclareMathOperator*{\Exp}{\ensuremath{{\mathsf{E}}}}
\DeclareMathOperator*{\poly}{\operatorname{poly}}

\newcommand{\calH}{\mathcal H}

\newcommand{\calU}{\mathcal U}

\newcommand{\alg}[1]{\ensuremath{\mathsf{#1}}\xspace}
\newcommand{\trycolor}{\alg{TryColor}}

\newcommand{\slackgeneration}{\alg{GenerateSlack}}
\newcommand{\multitrial}[1][]{\alg{MultiTrial}}
\newcommand{\slackcolor}[1][]{\alg{SlackColor}}

\newcommand{\matching}{\alg{Matching}}

\newcommand{\slicecolor}{\alg{SliceColor}}

\newcommand{\sct}{\alg{SynchColorTrial}}

\newcommand{\model}[1]{\ensuremath{\mathsf{#1}}\xspace}

\newcommand{\CONGEST}{\model{CONGEST}}

\newcommand{\congest}{\model{CONGEST}}
\newcommand{\local}{\model{LOCAL}}

\newcommand{\hK}{\uncolored{K}}

\newcommand{\Vsparse}{\ensuremath{V_{\mathsf{sparse}}}}

\newcommand{\pg}{\ensuremath{p_{\mathsf{g}}}}

\newcommand{\ce}{\ensuremath{C_{\mathsf{ext}}}}
\newcommand{\ca}{\ensuremath{C_{\mathsf{anti}}}}
\newcommand{\cslack}{\ensuremath{\gamma_{\mathsf{slack}}}}

\newcommand{\ov}[1]{\overline{#1}}
\newcommand{\avganti}{\ensuremath{\ov{a}}}
\newcommand{\avgext}{\ensuremath{\ov{e}}}

\newcommand{\spar}{\zeta}

\newcommand{\slack}{\spar}

\newcommand{\err}{\theta}
\newcommand{\errext}{\theta^{\mathsf{ext}}}
\newcommand{\erranti}{\theta^{\mathsf{anti}}}

\newcommand{\Ntwo}{N^{2}}

\newcommand{\td}{\widetilde{d}}
\newcommand{\pud}{\widetilde{\uncolored{d}}}

\newcommand{\te}{\widetilde{e}}
\newcommand{\ta}{\widetilde{a}}

\newcommand{\tA}{\widetilde{A}}

\newcommand{\rC}{\mathsf{C}}

\newcommand{\uncolored}[1]{#1^{\scriptscriptstyle\circ}} %% {\kern-0.5ex} ?
\newcommand{\colored}[1]{#1^{\scriptscriptstyle\bullet}}

\newcommand{\col}{\ensuremath{\mathcal{C}}}

\newcommand{\pal}[1]{\ensuremath{\Psi(#1)}}
\newcommand{\hatd}{\uncolored{d}}
\newcommand{\ID}{\ensuremath{\mathsf{ID}}\xspace}
\newcommand{\IDs}{\ensuremath{\mathsf{IDs}}\xspace}

\newcommand{\eqdef}{\stackrel{\text{\tiny\rm def}}{=}}
\newcommand{\evt}{\mathcal{E}}

\newcommand{\hashrep}{\calH^{\mathsf{rep}}}
\newcommand{\hashpwi}{\calH^{\mathsf{pwi}}}

\newcommand{\Kmod}{\mathcal{K}_{\mathsf{mod}}}
\newcommand{\Kvery}{\mathcal{K}_{\mathsf{very}}}

\title{Fast Coloring Despite Congested Relays}

\author{
Maxime Flin\\
\small Reykjavik University\\
\small \texttt{maximef@ru.is} \and
Magn\'us M. Halld\'orsson \\
\small Reykjavik University \\
\small \texttt{mmh@ru.is} \and
Alexandre Nolin\\
\small CISPA Helmholtz Center for Information Security\\
\small \texttt{alexandre.nolin@cispa.de}
}
\date{}

\begin{document}

\maketitle

\begin{abstract}
We provide a $O(\log^6 \log n)$-round randomized algorithm for distance-2 coloring in \congest with $\Delta^2+1$ colors.
For $\Delta\gg\poly\log n$, this improves exponentially on the $O(\log\Delta+\poly\log\log n)$ algorithm of [Halld\'orsson, Kuhn, Maus, Nolin, DISC'20].

Our study is motivated by the ubiquity and hardness of local reductions in \congest.
For instance, algorithms for the Local Lov\'asz Lemma [Moser, Tardos, JACM'10; Fischer, Ghaffari, DISC'17; Davies, SODA'23] usually assume communication on the conflict graph, which can be simulated in \local with only constant overhead, while this may be prohibitively expensive in \congest.
We hope our techniques help tackle in \congest other coloring problems defined by local relations.
%We hope our techniques help understanding harder coloring problems in \congest.
\end{abstract}

\newpage
\tableofcontents

\newpage

% \listoftodos

\section{Introduction}

    % Distributed models
% \paragraph{Distributed Models.}
In the \local model of distributed computing, we are given a communication network in the form of an $n$-node graph $G=(V,E)$, where each node has a unique $O(\log n)$-bit identifier.
Time is divided in discrete intervals called rounds, during which nodes send/receive one message to/from each of their neighbors in $G$.
In the \congest model, the same holds, and additionally each message is restricted to $O(\log n)$ bits.

The usual assumption in distributed graph coloring is that the communication graph $G=(V,E)$, through which messages are sent, is the same as the conflict graph $H$, the graph to be colored. When this assumption is loosened, $H$ is typically locally embedded into $G$,
%conflicts are typically local
in the sense that conflicting nodes (i.e., neighbors in $H$)
are simulated by nodes of $G$ within distance $c$ for some small constant $c > 0$.
% (as they are otherwise hard to handle even in \local)
The unlimited bandwidth of \local allows one
to simulate communication on $H$ with merely constant overhead in $G$.
%The lack of bandwidth constraint in \local allows
%to simulate in $G$ communication on $H$ with merely constant overhead.
%
In \congest however, bandwidth constraints preclude such local reductions.
This is a major challenge toward understanding the complexity landscape in \congest as such local transformations are ubiquitous in the distributed graph literature.
Notable examples include reductions to MIS \cite[Section 3.9]{barenboimelkin_book}, coloring algorithms based on the Lovász Local Lemma \cite{PS15,CPS17,CHLPU18}, or subroutines working with cluster graphs \cite{GGR20,MU21,FGGKR23}.

We make a step in that direction by studying the \emph{distance-2} $\Delta^2+1$-coloring problem, where $\Delta$ is the maximum degree of $G$.
Namely, we provide a fast algorithm for coloring $G^2$ while communicating on $G$ with $O(\log n)$-bit messages.
That is, each node $v\in V$, select a color in $\set{1, 2, \ldots, \Delta^2+1}$ different from the ones chosen by nodes at distance at most 2 from $v$ in $G$.
The main high-level challenges of distance-2 coloring (not knowing exact degrees nor colors used by neighbors) are somewhat similar to issues occurring naturally in harder coloring problems, such as estimating $c$-degrees (the number of neighbors with color $c$ in their palette) for all colors $c$ in parallel.

Coloring problems are amongst the most intensively studied problems in the distributed graph literature for they capture the main challenges of symmetry breaking (see, e.g., \cite{barenboimelkin_book}).
Symmetry breaking on power graphs appears naturally in numerous settings \cite{KMR01,BEPS,G16,G19,EM19,FGGKR23}.
(See, e.g.,~\cite[Section 1.2]{MPU23} for a recent treatment.)
For instance, it arises naturally when assigning frequencies to antennas in wireless networks.

\paragraph{Our Contribution.}
We provide a $\poly\log\log n$-round randomized algorithm to find a distance-2 coloring of $G$.
Our algorithm uses $\Delta^2+1$ colors, which is a natural analog to $\Delta+1$ at distance-1.

\begin{restatable}{theorem}{ddtheorem}\label{thm:d2}
There is a randomized algorithm for distance-2 coloring any $n$-node graph $G$ with maximum degree $\Delta$, using $\Delta^2+1$ colors, and running in $O(\log^6 \log n)$ rounds of \congest.
\end{restatable}

This is an exponential improvement over the previous best known bound $O(\log n)$ \cite{HKMN20}, as a function of $n$ alone.
Interestingly, for more general power graphs $G^k$ with $k\ge 3$, it is provably hard to verify an arbitrary coloring \cite{FHN20}.
Thus, any $\poly\log\log n$ algorithm coloring $G^k$ when $k \ge 3$ and $\Delta \gg \poly\log\log n$ would need a different approach.

\cref{thm:d2} requires non-constructive pseudorandom compression techniques, so can be viewed as either existential or requiring exponential local computation. However, we give an explicit and efficient algorithm that achieves such a coloring with $O(\log^2 n)$ bandwidth.
We emphasize that even with $O(\log^2 n)$ bandwidth, it is not clear that fast coloring algorithms can be implemented at distance-2.
Our $O(\log^2 n)$-bandwidth algorithm preserves the intuition behind our techniques.
In fact, reducing the bandwidth to $O(\log n)$ is a technical issue which was almost entirely solved by previous work \cite{HNT22}.

\subsection{Related Work}
% Distributed coloring: (The usual spiel)
% We consider the graph coloring problem in a communication-constrained distributed setting. Graph coloring is fundamental to distributed computation as an elegant way to resolve conflicts and break symmetry.
Coloring has been extensively studied in the distributed literature \cite{SW10,BEPS,barenboimelkin_book,HSS18,CLP20,HKNT22}, and it was the topic of the paper of Linial \cite{linial92} that defined the \local model.
% Distributed graph col
The best round complexity of randomized $(\Delta+1)$-coloring in \local (as a function of $n$ alone) progressed from $O(\log n)$ in the 80's \cite{luby86,alon86,johansson99}, through $O(\log^3 \log n)$ \cite{BEPS,HSS18,CLP20}, to the very recent $\widetilde{O}(\log^2\log n)$ \cite{GG23}.
These algorithms made heavy use of both the large bandwidth and the multiple-message transmission feature of the \local model.

In \congest, Halld\'orsson, Kuhn, Maus, and Tonoyan \cite{HKMT21} gave a $O(\log^5\log n)$-round \congest algorithm, later improved to $O(\log^3 \log n)$ in \cite{HNT22,GK21}.
Very recently, Flin, Ghaffari, Kuhn, and Nolin \cite{FGHKN23} provided a $O(\log^3\log n)$-round algorithm in \emph{broadcast} \congest, in which nodes are restricted to broadcast one $O(\log n)$-bit message per round.
While these algorithms drastically reduced the bandwidth requirements compared to their earlier \local and \congest counterparts, they still use more bandwidth than what distance-2 coloring allows.
Indeed, at distance-2 a node cannot receive a distinct message from each neighbor.

Recent years have seen several results for problems on power graphs in \congest \cite{GP19,HKM20,HKMN20,MPU23,BG23}.
Ghaffari and Portmann \cite{GP19} gave the first sublogarithmic network decomposition algorithm with a mild dependency on the size of identifiers.
Keeping a mild dependency on the size of identifiers is crucial in \congest as a common technique, called shattering, is to reduce the problem to small $\poly\log n$-size instances on which we run a deterministic algorithm, typically a network decomposition algorithm. While the instance size decreases exponentially, identifiers remain of size $O(\log n)$ bit. Hence, deterministic algorithms with linear dependency on the size of identifiers, such as \cite{RG20}, yield no sub-logarithmic algorithms.
The later $O(\log^5 n)$ \congest algorithm by \cite{GGR20} with mild dependency on the ID space was extended by \cite{MU21} to work on power graphs with exponentially large IDs space in time $O(\log^7 n)$.
Very recently, \cite{MPU23} gave a
%$\tilde{O}(k^2\log^5 n)$ algorithm to compute $k$-ruling sets on $G^k$.
% It resulted in a
$O(k^2\log\Delta\log\log n + k^4\log^5\log n)$ randomized \congest algorithm to compute a maximal independent set in $G^k$.
Along the way, \cite{MPU23} extended the faster $\widetilde{O}(\log^3 n)$ network decomposition of \cite{GGHIR23} to power graphs in \congest with a mild dependency on the ID space.

% Dist-2
For vertex-coloring at distance two, when $(1+\epsilon)\Delta^2$ colors are available, the problem is much easier, and is known to be solvable in $O(\log^4\log n)$ rounds \cite{HN23}.
The first $\poly(\log n)$-round \congest algorithm for distance-2 $(\Delta^2+1)$-coloring was given in \cite{HKM20}, while a $O(\log \Delta)+\poly(\log\log n)$-round algorithm was given in \cite{HKMN20}.
The original publication of this last result had a higher dependence in $n$, later reduced by improved network decomposition results of \cite{GGHIR23,MPU23} and a faster deterministic algorithm by \cite{GK21}. We state this for later use (see \cref{sec:appendix-small-degree} for more details):
\begin{restatable}
[{\cite[Lemma 3.12+3.15]{HKMN20}} + {\cite[Appendix A]{MPU23}} + \cite{GK21}]
{proposition}{lowDegColoring}
\label{lem:small-degree}
Let $H$ be a subgraph of $G^2$ where $G$ is the communication network, and suppose $\Delta(H)\le \poly\log n$.
Suppose each node $v$, of degree $d_H(v)$ in $H$, knows a list $L(v)$ of $d_H(v)+1$ colors from some color space $\card{\mathcal{U}} \le \poly(n)$.
There is a randomized algorithm coloring $H$ in $O(\log^5 \log n)$ rounds of \congest such that each node receives a color from its list.
\end{restatable}
The best bound known for a deterministic \CONGEST algorithm using $\Delta^2+1$ colors is $O(\Delta^2 + \log^* n)$ rounds \cite{HKM20}. Very recently, \cite{BG23} gave a handful of deterministic coloring algorithms on power graphs, including a $O(\Delta^4)$-coloring algorithm in $O(\log\Delta\cdot\log^* n)$ rounds (which is an adaptation of Linial's algorithm) and a $O(\Delta^2)$-coloring algorithm in $O(\Delta\log\Delta + \log\Delta \cdot\log^* n)$ rounds (which is an adaptation of the $O(\sqrt{\Delta}\poly\log\Delta + \log^* n)$ algorithm of \cite{FHK16,BEG_jacm22,MT22}).

\paragraph{Open Problems.}
Along the way we introduce various tools whose range of application extends to more general embedded coloring problems. In particular, almost all steps of our algorithm work if each $v\in V$ uses colors $\set{1, 2, \ldots, \td(v)+1}$, for $\td(v)$ a locally computable upper bound on degrees.
It would be interesting to know if $\td(v)+1$-coloring could be solved.
The more difficult list variants of this problem, where nodes must adopt colors from lists of size $\Delta^2+1$ or $\td(v)+1$, are also open.
Key aspects of our approach fail for list coloring and, in fact, it is not even known if $\Delta+1$-list-coloring of $G$ is achievable in $\poly\log\log n$ rounds of broadcast \congest.
It would also be interesting to push the complexity of $\Delta^2+1$-coloring of $G^2$ down to $O(\log^* n)$ when $\Delta \ge \poly\log n$, to fully match the state of the art at distance one.
While we conjecture that minor modifications to our algorithm\footnote{Such as reducing the nodes' uncolored degrees to $O(\log n/\log \log n)$ instead of $O(\log n)$, or slightly reducing the number of layers produced by one of the subroutines (\slicecolor).} might be able to reduce its complexity by one or more $\log \log n$ factors, achieving $O(\log^* n)$ should require a quite new approach, as many steps of our algorithm use $\Omega(\log \log \Delta)$ rounds.
% This question seems linked with the issue of non-constructiveness, which is not necessary at distance one.
%We hope our techniques can be extended to smaller palettes, the same way the $\Delta+1$-coloring algorithm of \cite{CLP20} was later extended to $\deg+1$-list-coloring by \cite{HKNT22}.
From a wider perspective, it would be very interesting to find precise characterizations of local embeddings that can be handled in \congest.

\subsection{Our Techniques in a Nutshell}

In this section, we highlight the main challenges and sketch the main ideas from our work.
Precise definitions are in \cref{sec:prelim} and a more detailed but still high-level overview of our algorithm can be found in \cref{sec:overview}.

\paragraph{Fast Distributed Coloring.}
All sublogarithmic distributed coloring algorithms \cite{HSS18,CLP20,HKMT21,HKNT22,HNT22,FGHKN23} follow the overall structure displayed in \cref{fig:flowchart}.
\begin{wrapfigure}{r}{0.3\textwidth}
%\begin{figure}
    \centering
    \includegraphics[width=0.295\textwidth]{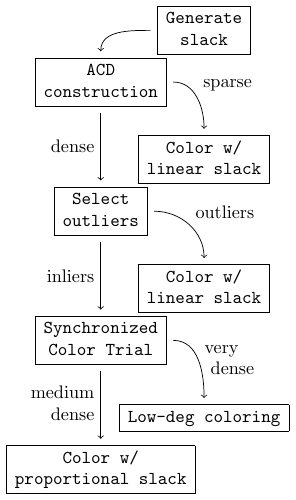}
    \caption{The structure of ultrafast coloring algorithms}
    \label{fig:flowchart}
%\end{figure}
\end{wrapfigure}
The key concept is the one of \emph{slack}: the slack of a node is the difference between the number of colors available to that node (i.e., not used by a colored neighbor) and its number of uncolored neighbors (see \cref{def:slack}).
Nodes with slack proportional to their uncolored degree can be colored fast. The algorithm uses a combination of creating excess colors (by coloring two neighbors with the same color) and reducing the uncolored degree in order for all nodes to get a slack linear in their uncolored degree.

We first generate \emph{slack} by a single-round randomized color trial. We next partition the nodes into the sparse nodes and dense clusters (called \emph{almost-cliques}). Among the dense clusters, we separate a fraction of the nodes as \emph{outliers}. Both the sparse nodes and the outliers can be colored fast using the \emph{linear} slack available to them. The remaining \emph{inliers} then go through a \emph{synchronized color trial} (SCT), where the nodes are assigned a near-uniform random color which avoids color clashes between nodes in the same cluster. The remaining nodes now should have slack proportional to their number of uncolored neighbors. The ultra-dense clusters need a special treatment, but they induce a low-degree graph. The above structure is necessary for high-degree graphs, while for low-degree graphs one can afford less structured methods.

At the outset, several parts of this schema already exist for distance-2 coloring. In particular, generating slack is trivial, a $\poly(\log\log n)$-round algorithm for coloring $\poly(\log n)$-degree graphs is known from \cite{HKMN20},
and coloring with slack $\Omega(\Delta^2)$ follows from \cite{HN23}.
Almost-clique decompositions (ACD) have been well studied and need only a minor tweak here.
We use a particularly simple form of SCT, introduced for the streaming and broadcast \congest settings \cite{FGHKN22,FGHKN23}: permute the colors of the \emph{clique palette}
-- the set of colors not used within the almost-clique -- and distribute them uniformly at random among the nodes.
We produce the clique palette by giving the nodes data structure access for looking up their assigned color.

\paragraph{Challenges.}
The biggest challenge in deriving efficient algorithms for the distance-2 setting is that there are no efficient methods known (and possibly not existing) to compute (or approximate) basic quantities like the distance-2 degree of a node. One can easily count the number of 2-paths from a given node, but not how many distinct endpoints they have. This seriously complicates the porting of all previous methods from the distance-1 setting, as we shall see.

A related issue is that a node cannot keep track of all the colors used by its distance-2 neighbors, since it has $\Delta^2$ of them but only bandwidth $O(\Delta \log n)$ bits. Hence, it cannot maintain its true palette (the set of available colors), which means that the standard method of coloring with proportional slack \cite{SW10,CLP20} (that can be achieved in $O(\log^* n)$ rounds in distance-1) is not available.

\paragraph{Using Multiple Sources of Slack.}
% In addition to the usual initial one-round slack generation,
We use slack from four sources in our analysis. The (usual) initial slack generation step gives the dense nodes slack proportional to their external degree -- their degree to outside of their almost-clique. The method of \emph{colorful matching} \cite{ACK19} provides slack proportional to the average anti-degree of the cluster, where anti-degree counts non-neighbors in one's almost-clique. And finally we get two types of slack for free: the discrepancy between the node's pseudodegree and its true degree on one hand, and the difference between $\Delta^2$ and the pseudodegree on the other hand, where pseudodegrees are easy-to-compute estimates of distance-2 degrees. Only by combining all four sources can we ensure that the final step of coloring with proportional slack can be achieved fast.

\paragraph{Selecting Outliers.}
The outliers are nodes with exceptionally high degree parameters, either high \emph{external degree} (to the outside of the almost-clique) or \emph{anti-degree} (non-neighbors within the cluster). As we cannot estimate their true values, we work with \emph{pseudodegrees}: the number of 2-paths to external neighbors, or how many additional 2-paths are necessary to reach all anti-neighbors.
The selection of outliers is crucial for the success of the last step of the algorithm, where we need to ensure that nodes have true slack proportional to the number of uncolored neighbors.
To select outliers, we use a sophisticated filtering technique, giving us bounds in terms of certain related parameters, that then can be linked to the slack that the nodes obtain.

\paragraph{Coloring Fast with Slack.}
With the right choice of inliers and suitable analysis of SCT, we argue that remaining uncolored nodes have slack proportional to their uncolored degree. We provide a new procedure to color these nodes, extending a method from the first fast \CONGEST algorithm \cite{HKMT21}. It needs to be adapted to biased sampling and to handle nodes with different ranges of slack. It outputs a series of low-degree graphs, which are then colored by the method of \cite{HKMN20}.

% \emph{Achieving $O(\log n)$ bandwidth:} Our method is simpler to explain when assuming $O(\log^2 n)$-bit bandwidth on the links. To reach the $O(\log n)$-bit bandwidth of \CONGEST, we utilize the pseudorandomness method of \emph{representative hash functions} of \cite{HNT22}.
% We also need to treat specially the case of ultradense cliques, where we are left with a low-degree graph after the SCT. In order to apply low-degree coloring algorithms, the nodes need to know their palettes. We derive an information dispersal method to achieve that.

% Finally, there are a number of steps that must be ported from the distance-1 case, with only moderate difficulty. We have already mentioned the colorful matching, introduced in \cite{ACK19} and applied in distributed coloring in \cite{FGHKN23}.
% The colorful matching is needed to combat false positives deriving from the use of the clique palette (i.e., it excludes the colors of the node's anti-neighbors, even though they may be valid for the node).
% Finally, the ACD construction of \cite{ACK19}, as imported to \CONGEST in \cite{HNT22}, turns out to work near-identically.

\subsection{Organization of the Paper}
After introducing some definitions and results from previous work in \cref{sec:prelim}, we give a detailed overview of the full algorithm in \cref{sec:overview}.
In \cref{sec:dense}, we go over the technical details involving the coloring of dense nodes, assuming a $O(\log^2 n)$ bandwidth.
We defer some lengthy and technical proofs to \cref{sec:slice-color,sec:outliers} to preserve the flow of the paper.
We explain how we reduce the bandwidth to $O(\log n)$ in \cref{sec:congestion}.
Missing details as well as some proofs can be found in appendices.
%%%%%%%%%%%%%%%%%%%%%%%%%

\section{Preliminaries \& Definitions}
\label{sec:prelim}

\paragraph{Distributed Graphs.}
For any integer $k \ge 1$, let $[k]$ represent the set $\{1, 2, \ldots, k\}$.
We denote by $G=(V, E)$ the communication network, $n = \card{V}$ its number of nodes, and $\Delta$ its maximum degree. The square graph $G^2$ has vertices $V$ and edges between pairs $u,v\in V$ if $\operatorname{dist}_G(u,v) \le 2$.
For a node $v\in V$, we denote its unique identifier by $\ID(v)$.
For a graph $H=(V_H,E_H)$, the neighborhood in $H$ of $v$ is $N_H(v)=\set{u\in V_H: uv\in E_H}$.
A subgraph $K=(V_K, E_K)$ of $H=(V_H,E_H)$ with $V_K \subseteq V_H$ is an \emph{induced} subgraph if $E_K = \set{uv\in E_H: u,v\in V_K}$, i.e., it contains all edges of $E_H$ between nodes of $V_K$.
We call \emph{anti-edge} in $H$ a pair $u,v\in V_H$ such that $uv\notin E_H$, i.e., an edge missing from $H$ (or, equivalently, in the complement of $H$).

The degree of $v$ in $H$ is $d_H(v)=|N_H(v)|$, and we shall denote by $\Ntwo(v)=N_{G^2}(v)$ the distance-2 neighbors of $v$, and the \emph{distance-2 degree} of $v$ by $d(v)=|\Ntwo(v)|$.
We also drop the subscript for distance-1 neighbors and write $N(v)$ for $N_G(v)$.

\paragraph{Distributed Coloring.}
A \emph{partial} $c$-coloring $\col$ is a function mapping vertices $V$ to colors $[c]\cup\set{\bot}$ such that if $uv\in E$, either $\col(u)\neq\col(v)$ or $\bot \in \set{\col(u), \col(v)}$.
The coloring is complete if $\col(v)\neq \bot$ for all $v\in V$ (i.e., all nodes have a color).
A $\deg+1$-list-coloring instance is an input graph $H=(V_H,E_H)$ where each node has a list $L(v)$ of $d_H(v)+1$ colors from some color space $\mathcal{U}$. A valid $\deg+1$-list-coloring is a proper coloring $\col:V_H\to \mathcal{U}$ such that $\col(v)\in L(v)$ for each $v\in V_H$.

Our algorithm computes a monotone sequence of partial colorings until all nodes are colored.
In particular, once a node \emph{adopts} a color, it never changes it.
The palette of $v$ with respect to the current partial coloring $\col$ is $\pal{v}\eqdef [\Delta^2+1]\setminus\col(\Ntwo(v))$, i.e., the set of colors that are not used by distance-2 neighbors.
% Initially $\pal{v}=\set{1, 2, \ldots, \td(v)+1}$.
For a set $S\subseteq V$, we shall denote the uncolored vertices of $S$ by $\uncolored{S}\eqdef\set{v\in S: \col(v)=\bot}$ and, reciprocally, the colored vertices of $S$ by $\colored{S} \eqdef S\setminus \uncolored{S}$.
We shall denote the uncolored (distance-2) degree with respect to $\col$ by $\hatd(v)\eqdef |\uncolored{N_{G^2}}(v)|$.

\subsection{Slack Generation}

A key notion to all fast randomized coloring algorithm is the one of slack.
It captures the number of excess colors: a node with slack $s$ will always have $s$ available colors, regardless of the colors tried concurrently by neighbors.
For our problem, the slack is more simply captured by the following definition.

% \begin{definition}[Slack]
% For any (partial) coloring $\col$ of $G$, we say \emph{$v$ has slack $s$ w.r.t.\ $\col$} if for any complete coloring extending $\col$ there are always at least $s$ colors in $\pal{v}$.
% \end{definition}

\begin{definition}[Slack]\label{def:slack}
Let $H$ be an induced subgraph of $G^2$. The slack of $v$ in $H$ (with respect to the current coloring of $G^2$) is
\[ s_H(v) \eqdef |\pal{v}|-\hatd_H(v) \ . \]
\end{definition}

There are three ways a node can receive slack: if it has a small degree originally, if two neighbors adopt the same color, or if an uncolored neighbor is inactive (does not belong to $H$). We consider the first two types of slack \emph{permanent} because a node never increases its degree, and nodes never change their adopted color. On the other hand, the last type of slack is \emph{temporary}: if some inactive neighbors become active, the node loses the slack which was provided by those neighbors.

The sparsity of a node counts the number of missing edges in its neighborhood.
We stress that, contrary to previous work in $\Delta+1$-coloring \cite{CLP20,HKMT21,FGHKN23}, we use the \emph{local sparsity} -- defined in terms of the node's degree $d(v)$ -- as opposed to the global sparsity, instead defined in term of $\Delta$.
This is to separate the contribution to slack of same-colored neighbors from the \emph{degree slack}, $\Delta^2-d(v)$.
While global sparsity measures both, local sparsity focuses on the former.
%We use the local sparsity, which does not count degree slack, to clearly separate contributions to slack.

\begin{definition}[Local Sparsity, \cite{AA20,HKNT22}]
\label{def:sparsity}
The sparsity of $v$ (in the square graph $G^2$) is
% \[ \spar_v \eqdef \frac{1}{d(v)}\parens*{{d(v) \choose 2} - |E(\Ntwo(v))|} \ . \]
\[ \spar_v \eqdef \frac{1}{d(v)}\parens*{\binom{d(v)}{2} - |E(\Ntwo(v))|} \ . \]
A node $v$ is $\spar$-sparse if $\spar_v \ge \spar$; if $\spar_v \le \spar$ it is $\spar$-dense.
\end{definition}

For a node $v$, observe that each time that both endpoints of a missing edge in $N(v)$ are colored the same, the node $v$ gains slack as its uncolored degree decreases by 2 while its palette loses only 1 color. Therefore, when a node has many missing edges in its neighborhood, it has the potential to gain a lot of slack \cite{Reed98,EPS15}. % This was first shown in \cite{EPS15}.
This potential for slack is turned into permanent slack by the following simple algorithm (\slackgeneration): each node flips a random coin (possibly with constant bias); each node whose coin flip turned heads picks a color at random and \emph{tries} it, i.e., colors itself with it if none of its neighbors is also trying it (see \cref{sec:pseudo-code} for pseudo-code).
As we state the result with local sparsity (which is in terms of $d(v)$) while nodes try colors in $[\Delta^2+1]$, the next statement has a $d(v)/\Delta^2$ factor compared to previously published versions.

\begin{proposition}[Slack Generation, \cite{Reed98,EPS15,HKMT21}]
\label{lem:slack-generation}
There exists a (small) universal constant $\cslack> 0$ such that after \slackgeneration, w.p.\ $e^{-\Omega(\slack_v \cdot \frac{d(v)}{\Delta^2})}$, node $v$ received slack $\cslack\cdot \slack_v\cdot \frac{d(v)}{\Delta^2}$.
\end{proposition}

% \subsection{(Local) Sparse-Dense Decomposition}
\subsection{Sparse-Dense Decomposition}
\label{sec:acd-explain}

All recent fast randomized distributed coloring algorithms \cite{HSS18,CLP20,HKMT21,HKNT22,FHM23,FGHKN23} decompose the graph into a set of sparse nodes and several dense clusters.
Such a decomposition was first introduced by \cite{Reed98}.
% We use the \emph{local} version of \cite{AA20} which allows clusters of different size for nodes of different degrees.

\begin{definition}
\label{def:almost-clique}
For $\epsilon \in (0, 1/3)$, a distance-2 $\epsilon$-almost-clique decomposition (ACD) is a partition of $V(G)$ in sets $\Vsparse, K_1, \ldots, K_k$ such that
\begin{enumerate}
\item nodes in $\Vsparse$ either are $\Omega(\epsilon^2\Delta^2)$-sparse in $G^2$ or have degree $d(v)\leq \Delta^2 - \Omega(\epsilon^2\Delta^2)$,
% \item nodes in $\Vuneven$ are $\Omega(\epsilon^2\Delta)$-uneven,
\item for all $i\in[k]$, sets $K_i$ are called \emph{almost-cliques}, and verify
% let $\Delta_{K_i}=\max_{v\in K_i} d(v)$ be the maximum degree in almost-clique $K_i$,
\begin{enumerate}
\item\label[part]{part:clique-size} $|K_i| \le (1+\epsilon)\Delta^2$,
\item\label[part]{part:deg-inside} for each $v \in K_i$, $|\Ntwo(v)\cap K_i|\ge (1-\epsilon)\Delta^2$.
\end{enumerate}
\end{enumerate}
\end{definition}

There are several ways to compute this decomposition in \congest \cite{HKMN20,HKMT21,HNT22,FGHKN22}.
We refer the reader to the version of \cite[Section 4.2]{HNT22}. The existing distance-2 algorithm of \cite{HKMN20} uses $O(\log\Delta)$ rounds and the \congest algorithms by \cite{HKMT21} require too much bandwidth at distance-2.
We mention that \cite{FGHKN22} implements \cite{HNT22} without representative hash functions and that it can be done here as well.
We discuss the implementation of the algorithm in \cref{sec:acd}.

\begin{lemma}[{Adaptation of \cite[Section B.1]{FGHKN22}}]
\label{lem:acd}
There exists a \congest randomized algorithm partitioning the graph into $\Vsparse, K_1, \ldots, K_k$ for some integer $k \ge 0$ such as described in \cref{def:almost-clique}. It runs in $O(\epsilon^{-4})$ rounds.
\end{lemma}

\begin{definition}[External and Anti-Degrees]
For a node $v\in K$ and some almost-clique $K$,
we call $e_v = |\Ntwo(v)\setminus K|$ its \emph{external degree} and $a_v=|K\setminus \Ntwo(v)|$ its anti-degree.
We shall denote by
$\avgext_K=\sum_{v\in C} e_v/|K|$ the average external degree and
$\avganti_K=\sum_{v\in K}a_v/|K|$ the average anti-degree.
\end{definition}

It was first observed by \cite{HKMT21} that sparsity bounds external and anti-degrees.
% \cite{HKNT22} extended this to $\deg+1$-list-coloring and bounded these quantities by $\spar + \discr$.

\begin{lemma}[{\cite[Lemmas 6.2]{HKMT21}}]
\label{lem:bound-ext}
\label{lem:bound-anti}
There exists two constants $\ce = \ce(\epsilon) > 0$ and $\ca=\ca(\epsilon) > 0$ such that for all $v\in K$, the bounds $e_v \le \ce \slack_v$ and $a_v \le \ca \slack_v$ holds.
\end{lemma}

\subsection{Pseudo-degrees}
Bandwidth constraints, such as that of the \CONGEST model, can severely restrict nodes in their ability to learn information about their neighborhood in a power graph of $G$. This includes a node's palette (which colors are not yet used by its neighbors in the power graph) but also its degree and related quantities. This motivates the use of similar, but readily computable quantities.

\begin{definition}[Distance-2 Pseudo-Degrees]
\label{def:pseudodegree}
In the distance-2 setting, for any node $v\in V$, let its \emph{pseudo-degree} $\td(v)$ and its \emph{uncolored pseudo-degree} $\pud(v)$ be
\begin{equation}
\td(v) \eqdef \sum_{\mathclap{u\in N_G(v)}} |N_{G}(u)|
\quad\text{and}\quad
\pud(v)\eqdef |\uncolored{N}_G(v)| + \sum_{\mathclap{u\in N_G(v)}} |\uncolored{N}_{G\setminus \set{v}}(u)|\ .
\label{eq:pseudodegree}
\end{equation}
For a dense node $v\in K$, its \emph{pseudo-external degree} $\te_v$ and its \emph{pseudo-anti degree} $\ta_v$ are
\begin{equation}
\te_v \eqdef \sum_{\mathclap{u\in N_G(v)}} |N_G(u) \setminus K|
\quad\text{and}\quad
\ta_v \eqdef |K| - \sum_{\mathclap{u\in N_G(v)}} |N_G(u)\cap K|\ .
\label{eq:pseudoantiext}
\end{equation}
\end{definition}

Note that pseudo-degree and pseudo-external degree are \emph{overestimates} of a node's actual degree and external degree, while pseudo-anti degree is an \emph{underestimate} of a dense node's actual anti-degree. The estimates are accurate for nodes with a tree-like 2-hop neighborhood.

For dense nodes, we also introduce notation for the deviations between the pseudo-degrees and actual $G^2$-degrees. Such deviations result in slack, which we exploit later in the paper.
\[ \errext_v \eqdef \te_v - e_v\ , \quad \erranti_v \eqdef a_v - \ta_v \ , \quad \text{and} \quad \err_v \eqdef \errext_v + \erranti_v = \td(v) - d(v) \ . \]
We also write $\errext_K = \sum_{v\in K} \errext_v/|K|$ for the average value within a clique.

Pseudo-degrees partially allow nodes to estimate their \emph{degree slack}, the number of colors that $v$ is guaranteed to always have available due to the palette being larger than its degree. Intuitively, the deviations $\errext_v$ and $\erranti_v$ capture the part of its degree slack that a dense node $v$ does not know about.
\begin{equation}
    \label{eq:knowndegrereslack}
    \underbrace{\Delta^2+1 - d(v)}_{\mathclap{\text{degree slack}}}
    \;=\;
    \underbrace{\Delta^2+1 - \td(v)}_{\mathclap{\text{known to $v$}}}%
    \;+\;%
    \underbrace{ \errext_v + \erranti_v}_{\mathclap{\text{unknown to $v$}}}
\end{equation}

\section{Detailed Overview of the Full Algorithm}
\label{sec:overview}

% As our proof is lengthy and technical, we highlight here the main new ideas and give a streamlined overview of how they come together.
We now give a streamlined overview of our algorithm and describe with some details the technical ideas behind it.
% Our algorithm follows the same framework as previous work.
See \cref{alg:high-level} for a high-level description of its steps.
Since there exists a $O(\log^5 \log n)$-round algorithm when $\Delta \le \poly\log n$ (\cref{lem:small-degree}), we assume $\Delta \ge \Omega(\log^{3.5} n)$.
Henceforth, we assume we are given the almost-clique decomposition $\Vsparse, K_1, \ldots, K_k$ (\cref{lem:acd}).

\begin{algorithm}
\caption{High-Level Algorithm}
\label{alg:high-level}
\Input{Graph $G$ with $\Delta \ge \Omega(\log^{3.5} n)$}
\Output{A distance-2 coloring $\col$ of $G$}

$\Vsparse, K_1, \ldots, K_k$ = \alg{ComputeACD}{}$(\epsilon)$ \label{step:acd}\hfill (\cref{sec:acd-explain})

\slackgeneration \label{step:slackgen}\hfill(\cref{lem:slack-generation})

\alg{ColoringSparseNodes} \label{step:sparse}\hfill (\cref{lem:coloring-sparse})

\matching \label{step:matching}\hfill (\cref{sec:colorful-matching})

\alg{ComputeOutliers} \label{step:outliers} \hfill (\cref{sec:outliers})

\alg{ColorOutliers} \label{step:color-outliers} \hfill (\cref{lem:coloring-sparse})

\sct \label{step:sct} \hfill (\cref{sec:sct})

$L_1, \ldots, L_{\ell}=$ \slicecolor for some $\ell=O(\log\log n)$ \label{step:slice-color} \hfill (\cref{sec:slack-color-log2n,sec:slice-color})

\ForEach{$i\in[\ell]$}{
\alg{LearnPalette}  \label{step:learn-palette}\hfill (\cref{sec:learn-palette-log2n})

\alg{ColorSmallDegree}{}$(L_i)$ \label{step:small-degree}\hfill (\cref{lem:small-degree})
}
\end{algorithm}

\paragraph{Coloring Sparse Nodes (Steps \ref{step:slackgen} \& \ref{step:sparse}).}
% Given the almost-clique decomposition, we can color the sparse nodes in $O(\log^* n)$ rounds.
The coloring of sparse nodes was already handled in \cite{HN23}.
After \slackgeneration, all sparse nodes have slack proportional to $\Delta^2$ (\cref{lem:slack-generation}).
In particular, their palettes always represent a constant fraction of the color space $[\Delta^2+1]$. This allows them to sample colors in their palette efficiently without learning most of their distance-2 neighbors' colors.
The algorithm is summarized by the following proposition:

\begin{restatable}[Coloring Nodes with Slack Linear in $\Delta^2$, \cite{HN23}]{proposition}{propColoringSparse}
\label{lem:coloring-sparse}
Suppose $\Delta \ge \Omega(\log^{3.5} n)$.
Let $H$ be an induced subgraph of $G^2$ for which all nodes have slack $\gamma \cdot \Delta^2$ for some universal constant $\gamma > 0$ known to all nodes.
There exists an algorithm coloring all nodes of $H$ in $O(\log^* n)$ rounds.
\end{restatable}

\paragraph{Reducing Degrees with Slack.}
Since coloring sparse nodes is already known, from now on, we focus our attention to dense nodes.
Reducing coloring problems to low-degree instances that one then solves with an algorithm that benefits from the low degree is a common scheme in randomized algorithms for distributed coloring \cite{BEPS,CLP20}.
In particular, when nodes have slack linear in their degree, it was observed by \cite{SW10,EPS15,CLP20} that if nodes try \emph{multiple colors from their palette}, degrees decrease exponentially fast, resulting in a $O(\log^* n)$-round algorithm in \local.
This observation motivates the structure of all ultrafast coloring algorithms: 1) generate $\Omega(e_v)$ slack with \slackgeneration, 2) reduce degrees to $O(e_v)$ with \sct, and 3) complete the coloring with slack.
% For example, in \cite{BEPS} nodes try $O(\log \Delta)$ random colors, one at a time, to obtain two subgraphs of degree $O(\log n)$ to solve, and algorithms that use multiple color trials \cite{SW10,CLP20} try $\Theta(\log n)$ colors over $O(\log^* n)$ rounds to obtain an instance of maximum degree $O(\log n)$.
% Neither approach works for us, as \cite{BEPS} is too slow for our target complexity and multiple color trials would need higher bandwidth than we have access to in our models.
Unfortunately, this approach is not feasible for us because it requires too much bandwidth.
As a result, we do something intermediate that takes advantage of slack
% as the multiple color trials do,
but only tries a single color at a time to accommodate our bandwidth limitations.
In $O(\log\log n)$ rounds, our method creates $O(\log \log n)$ instances of the maximum degree $O(\log n)$.

Another key technical detail of these methods is that nodes try colors \emph{from their palettes}.
At distance-2, perfect sampling in one's palette is not feasible for nodes do not have sufficient bandwidth.
We show that they can nevertheless sample colors from a good enough approximation of their palette, in the sense that it preserves the slack.
Our involved sampling process requires our degree reduction algorithm to work with weaker guarantees than previous work \cite{HKMT21,HKNT22}. See \cref{sec:slice-color} for the proof.

\begin{restatable}[Slice Color]{lemma}{lemmaSlackColor}
\label{lem:slack-color}
Let $C, \alpha, \kappa > 0$ be some universal constants.
Suppose each node knows an upper bound $b(v) \geq \hatd(v)$ on its uncolored degree. Suppose that for all nodes with $b(v) \geq C\log n$, and a value $s(v) \ge \alpha \cdot b(v)$, there exists an algorithm that samples a color $\rC_v\in \pal{v}\cup\set{\bot}$ (where $\bot$ represents failure) with the following properties:
%such that , and there exists a universal constant $\kappa>0$ such that for all colors $c\in \pal{v}$, we have
\begin{align}
\Pr(\rC_v = \bot) &\le 1/\poly(n)\ , \\
\Pr(\rC_v = c \mid \rC_v\neq\bot) &\le \frac{\kappa}{\hatd(v)+s(v)} \ .
\label{eq:uniform-slack-color}
\end{align}
Then, there is a $O(\log \log \Delta + \kappa\cdot\log(\kappa/\alpha))$-rounds algorithm extending the current partial coloring so that uncolored vertices are partitioned into $\ell=O(\log\log \Delta)$ layers $L_1, \ldots, L_{\ell}$ such that each uncolored node knows to which layer it belongs and each $G[L_i]$ has uncolored degree $O(\log n)$.
\end{restatable}

\paragraph{Coloring Dense Nodes.}
We assume the sparse nodes are colored (Step \ref{step:sparse}) and focus on the dense nodes (Steps \ref{step:matching} to \ref{step:small-degree}).
Dense nodes receive slack proportional to their external degree (Step \ref{step:slackgen}, \cref{lem:slack-generation,lem:bound-ext}) in all but the densest almost-cliques.

\paragraph{Steps \ref{step:matching}, \ref{step:outliers} \& \ref{step:color-outliers}: Setting up (\cref{sec:setup}).}
We begin by two pre-processing steps to ensure uncolored nodes have useful properties further in the algorithm.
Computing a colorful matching (Step \ref{step:matching}, \cref{lem:colorful-matching}) creates $\Theta(\avganti_K)$ slack in the clique palette $\pal{K}$.
%, which is the set of colors not used by nodes of $K$.
This is a crucial step to ensure we can approximate nodes palettes (see Step~\ref{step:slice-color}).
We then compute a (small) fraction $O_K \subseteq K$ of atypical nodes called \emph{outliers} (Step \ref{step:outliers}, \cref{lem:compute-outliers}).
Outliers have $\Omega(|K|)$ slack from their inactive inlier neighbors, and can thus %; hence, are essentially sparse and can
be colored in $O(\log^* n)$ rounds (Step \ref{step:color-outliers}, \cref{lem:coloring-sparse}).
\emph{Inliers} $I_K \eqdef K\setminus O_K$ verify $\te_v \le O(\avgext_K + \errext_K)$ and $\ta_v\le O(\avganti_K)$ (\cref{eq:inliers}).

While the colorful matching algorithm is rather straightforward to implement even at distance-2, computing outliers is a surprisingly challenging task.
The reason is that, contrary to distance-1, nodes do not know good estimates for $\avganti_K$.
Fortunately, a node only overestimates its anti-degree (i.e., $\ta_v \ge a_v$) and we know that $0.9|K|$ nodes have $a_v \le 100\avganti_K$.
By learning approximately the distributions of anti-degrees, we can set a threshold $\tau$ such that all nodes with $\ta_v \le \tau$ verify $\ta_v \le 200\avganti_K$.

\paragraph{Step \ref{step:sct}: Synchronized Color Trial (\cref{sec:sct}).}
This now standard step exploits the small external degree of dense nodes to color most of them.
We distributively sample a permutation $\pi$ of $[|I_K|]$ such that the $i$-th node in $I_K$ (with respect to any arbitrary order) knows $\pi(i)$ (\cref{lem:perm}).
Each node then learns the $\pi(i)$-th color in $\pal{K}$ and tries that color (\cref{lem:free-color}).
This leaves $O(\avganti_K+\avgext_K+\log n)$ uncolored node in each almost-clique (\cref{lem:sct}).
To implement these steps, we split nodes into small random groups to spread the workload.
The main technical novelty here is an algorithm to aggregate the partial information of each group (\cref{lem:prefix-sum}).

\paragraph{Step \ref{step:slice-color}: Slice Color (\cref{sec:slack-color-log2n}).}
In the densest almost-cliques, the synchronized color trial already leaves $O(\log n)$ nodes uncolored with uncolored degree $O(\log n)$.
In other almost-cliques, nodes have slack proportional to their uncolored degree $O(\avganti_K +\avgext_K + \te_v)$: $\Theta(\avganti_K)$ slack from the colorful matching (Step \ref{step:matching}), $\Omega(e_v)+\errext_v=\Omega(\te_v)$ from slack generation and pseudo-external-degree. If these are not large enough, it must be that $\Delta^2-\td(v) \ge \Omega(\avgext_K)$, i.e., the node has enough slack from its small degree.

While at distance-1 we could use slack to color fast, doing the same at distance-2 requires more work because nodes do not know their palettes.
The first key observation, is that \emph{the clique-palette preserves the slack}. More precisely, for all inliers $v\in I_K$, we have $|\pal{K}\cap\pal{v}| \ge \hatd(v) + \Omega(\avganti_K + \avgext_K + \te_v)$ (\cref{lem:slack}).
The proof of this statement is very technical and requires careful balancing of all four sources of slack: the colorful matching, the sparsity slack, the pseudo-degree slack and the degree slack.
We also emphasize this is why inliers need to verify $\ta_v \le O(\avganti_K)$: when we use colors from the clique-palette, we lose up to $a_v$ colors used by anti-neighbors, which we compensate using the colorful matching and pseudo-anti-degree slack $\Theta(\avganti_K) + \erranti_K$.

It remains to sample uniform colors in $\pal{K}\cap\pal{v}$.
Based on \cref{lem:slack}, it can be observed that $|\pal{K}\cap\pal{v}| \ge \Omega(|\pal{K}|)$.
Hence, each node $v$ finds a random color $\rC_v\in\pal{K}\cap\pal{v}$ to try w.h.p.\ by sampling $\Theta(\log n)$ uniform colors in $\pal{K}$.
With $\Theta(\log^2 n)$ bandwidth, this step can easily be implemented (\cref{lem:sample-log2n}) by sampling indices in $|\pal{K}|$ and using the same tools as for the synchronized color trial (Step \ref{step:sct}).
With $\Theta(\log n)$ bandwidth, we use \emph{representative hash functions} (\cref{lem:learn-moderately-logn}).
Intuitively, we use a $\poly(n)$-sized family of hash functions mapping $[\Delta^2+1]$ to some $[\Theta(|\pal{K}|)]$. To ``sample'' colors, we take a hash function at random and pick as sampled colors those hashing below $\Theta(\log n)$.
Since a hash function $h$ can be described in $O(\log n)$ bits and the hashes $h(\pal{K})\cap[\sigma]$ and $h(\col(N(v)\setminus K))\cap[\sigma]$ can be described using a $O(\log n)$-bitmap, the algorithm works in \congest.

The above allows us to apply \slicecolor after \sct. In $O(\log\log n)$ rounds, we compute $\ell=O(\log\log n)$ layers $L_1, L_2, \ldots, L_{\ell}$ such that the maximum uncolored degree in each induced graph $G[L_i]$ is $O(\log n)$ (\cref{lem:slack-color}).

\paragraph{Steps \ref{step:learn-palette} \& \ref{step:small-degree}: Coloring Small Degree Instances (\cref{sec:learn-palette-log2n}).}
We go through each layer $L_1, L_2, \ldots, L_{\ell}$ sequentially, each time coloring all nodes in $L_i$.
Actually constructing small degree instances for solving with a deterministic algorithm requires the nodes to learn colors from their palette -- a tough ordeal in the distance-2 setting.
Our argument is two-fold: in not-too-dense almost-cliques, a simple sampling argument works (\cref{lem:learn-moderatly-dense-log2n}).
In very dense almost-clique where $\avganti_K, \avgext_K, \errext_K \le O(\log n)$, we use a different argument exploiting the very high density of the cluster to disseminate colors fast (\cref{lem:learn-palette-very-dense-log2n}).
We point out that at this step, it is crucial that uncolored nodes have typical degrees $a_v, e_v \le O(\log n)$, which is ensured by our inlier selection (Step \ref{step:outliers}).
Once nodes know a list of $\hatd(v)+1$ colors from their palettes, we can use a small-degree algorithm from \cite{HKMN20,GGHIR23,MPU23,GK21} to complete the coloring of $L_i$ in $O(\log^5 \log n)$ rounds (\cref{lem:small-degree}).
Overall, coloring small degree instances needs $O(\log^6 \log n)$ rounds, which dominates the complexity of our algorithm.

\section{Coloring Dense Nodes}
\label{sec:dense}

Henceforth, we assume we are given an $\epsilon$-almost-clique decomposition $\Vsparse, K_1, \ldots, K_k$ (for some small\footnote{we made no attempt to optimize the constants} $\epsilon$, e.g., $\epsilon=10^{-5}$) where $\Vsparse$ is already colored.
We further assume we ran \slackgeneration and that each node $v$ with $\slack_v \ge \Omega(\log n)$ has slack $\Omega(\slack_v)$ (\cref{lem:slack-generation}).
In this section, we describe an algorithm that colors dense nodes.
% reduces the
% problem
% % $\Delta^2+1$-coloring of $G^2$
% to $\deg+1$-list-coloring instances with maximum degree $O(\log n)$.
More formally, we prove the following result:

\begin{restatable}[Coloring Dense Nodes]{proposition}{propColoringDenseNodes}
\label{lem:coloring-dense}
After \slackgeneration and coloring sparse nodes, there is a $O(\log^6\log n)$-round randomized algorithm $\Delta^2+1$-coloring dense nodes $K_1, \ldots, K_k$ with high probability.
% reducing  the remaining $\Delta^2+1$-coloring instance on dense nodes $K_1, \ldots, K_k$
% to $O(\log\log n)$ $\deg+1$-\emph{list}-coloring instances, each with maximum degree $O(\log n)$ and where each node knows its list.
\end{restatable}

% We emphasize that at the end of the algorithm, each node knows its list of colors. While this is trivial at distance-1, achieving this requires some work at distance-2 (recall nodes do not know colors of their neighbors).

We find it helpful to describe the algorithm using $O(\log^2 n)$ bandwidth first, for it is based on the same high-level ideas.
The use of extra bandwidth is very limited and explicitly stated.
The reduction in bandwidth is mostly achieved through techniques from \cite{HNT22}.
Details of how this is performed can be found in \cref{sec:congestion}.

\subsection{Leader, Outliers \& Colorful Matching}
\label{sec:setup}

A useful property of almost-cliques, used by \cite{HKMT21,HKNT22,FGHKN23}, is their relative uniform sparsity.
The first step of these algorithms is to dissociate the typical nodes, called \emph{inliers}, from the atypical ones, called \emph{outliers}.
At distance-2, however, detecting outliers is difficult.
%Contrary to previous work, at distance-2 detecting outliers is more challenging.
For instance, the algorithm of \cite{FGHKN23} requires to keep only nodes with anti-degree $a_v \le O(\avganti_K)$.
Such a trivial task at distance one requires work at distance-2 because nodes are unable to approximate their degree accurately (up to a constant factor).
To circumvent this limitation of the distance-2 setting, we instead compute outliers using \emph{pseudo-degrees} (\cref{def:pseudodegree}).

\begin{restatable}[Compute Outliers]{lemma}{lemmaOutliers}
\label{lem:compute-outliers}
We compute in $O(\log\log\Delta)$ rounds a set $O_K$ in each almost-clique $K$ such that $I_K \eqdef K\setminus O_K$ has size $0.95|K|$ and each $v\in I_K$ verifies
\begin{equation}
\label{eq:inliers}
\te_v \le 200(\avgext_K + \errext_K) \ , \quad \text{and}\quad \ta_v \le 200\avganti_K \ .
\end{equation}
\end{restatable}

The general idea behind \cref{lem:compute-outliers} is that a large fraction of the almost-clique has a typical sparsity, external degree and anti-degree. By learning approximately the distribution of pseudo-external degrees and pseudo-anti-degrees, the leader can select a large enough fraction of $K$ verifying \cref{eq:inliers}.
% We insist, however, that we may select nodes $v$ with
% $e_v \gg \avgext_K$
% % a $\spar_v \gg \slack_K$
% or $a_v \gg \avganti_K$.
% The former is not an issue as nodes have $\Omega(e_v)$ slack; we compensate for the error in the latter using degree slack.
As the proof of \cref{lem:compute-outliers} is quite technical, we defer it to \cref{sec:outliers} to preserve the flow of the paper.

Outliers can be colored in $O(\log^* n)$ rounds, thanks to the $\Omega(\Delta^2)$ slack provided by their inactive inliers neighbors. Starting from \cref{sec:sct}, we will assume outliers are all colored, thus focus on coloring inliers.

\paragraph{Colorful Matching.}
A major issue when coloring dense nodes in $G^2$ is that they do not know their palette.
We overcome this by using the \emph{clique palette} as an approximation.

\begin{definition}[Clique Palette]
\label{def:clique-palette}
% For an almost-clique $K$, define its \emph{clique palette} as $\pal{K}=[\td(w_K)]\setminus \col(K)$, i.e., the set of colors in $\set{1, 2, \ldots, \td(w_K)}$ that are not already used by a node of $K$.
For an almost-clique $K$, define its \emph{clique palette} as $\pal{K}=[\Delta^2+1]\setminus \col(K)$, i.e., the set of colors in $\set{1, 2, \ldots, \Delta^2+1}$ that are not already used by a node of $K$.
\end{definition}

This idea was first (implicitly) used by \cite{ACK19} to prove their palette sparsification theorem on almost-cliques.
This was since used formally in \cite{FGHKN22,FGHKN23}.
Note that in large almost-cliques (such that $|K| = (1+\epsilon)\Delta^2$), the clique-palette can be empty after coloring the outliers.
To remedy this issue, \cite{ACK19} compute first a colorful-matching:

\begin{definition}[Colorful Matching]
\label{def:colorful-matching}
In a clique $K$, a colorful matching $M$ is a set of anti-edges in $K$ (edges in the complement) such that both endpoints are colored the same.
\end{definition}

In an almost-clique with a colorful matching of size $|M|$, the slack provided to the clique palette by monochromatic anti-edges ensures the following fact:

\begin{fact}
\label{fact:simple-almost-clique}
In an almost-clique $K$ with $k$ uncolored nodes and a colorful matching of size $\avganti_K$,
we have $|\pal{K}| \ge k$.
\end{fact}

We will use a stronger variant of \cref{fact:simple-almost-clique} which requires a colorful matching of size $\beta\avganti_K$ for some large constant $\beta > 1$ (to be defined later).
\cite{FGHKN23} gave a \congest algorithm to compute a colorful matching of size $\Theta(\avganti_K/\epsilon)$ in $O(1/\epsilon)$ rounds in cliques with a high average anti-degree.
We review this algorithm and argue it can be implemented on $G^2$ with constant overhead in \cref{sec:colorful-matching}.

\begin{proposition}[Distance-2 Colorful Matching]
\label{lem:colorful-matching}
Let $\beta \le O(1/\epsilon)$.
There exists a $O(\beta)$-round randomized algorithm \matching that computes a colorful matching of size $\beta\avganti_K$ in all almost-cliques of $G^2$ with $\avganti_K \ge \Omega(\log n)$.
\end{proposition}

\subsection{Synchronized Color Trial}
\label{sec:sct}

Synchronizing color trials in dense components is a fundamental part of all known sub-logarithmic algorithm \cite{HSS18,CLP20,HKMT21,HKNT22}.
We implement a variant of \cite{HKNT22} where a uniform permutation determines which node tries which color.
Contrary to \cite{HKNT22}, we use colors from the clique palette $\pal{K}$ (\cref{def:clique-palette}), which is easier to implement in our setting.
This approach was also used by \cite{FGHKN23} to implement the synchronized color trial in Broadcast-\CONGEST.
A major difference with \cite{FGHKN23} is that at distance-2, nodes cannot learn the whole clique-palette $\pal{K}$.
% Also note that, due to errors when approximating degrees, \cref{lem:sct} states a weaker variant of \cite{HKNT22} in the sense that $O(\errext_K)$ more node remain uncolored after the synchronized color trial.

\begin{lemma}[Synchronized Color Trial, \cite{HKNT22}]
\label{lem:sct}
Let $K$ be an almost-clique with $|I_K|\ge \Omega(|K|)$ inliers.
Fix the randomness outside $K$ arbitrarily.
Let $\pi$ be a uniform random permutation of $[|I_K|]$.
If the $i$-th node in $I_K$ (for any arbitrary order) tries the $\pi(i)$-th color in $\pal{K}$ (if it exists), then, with high probability, at most $O(\avgext_K + \avganti_K + \log n)$ are uncolored in $K$.
\end{lemma}

\begin{proof}
Let $v\in K$ be the node of minimum anti-degree.
Note that $|K| \le \Delta^2 + a_v \le \Delta^2 + \avganti_K$.
Since all uncolored nodes are inliers, even if each colored node in $K$ uses a different color, the clique palette has size $\pal{K} \ge \Delta^2-(|K|-|I_K|) \ge |I_K| - \avganti_K$.
This means that at most $\avganti_K$ nodes fail due to colors missing in $\pal{K}$.
% As we explain further, nodes can compute $|\pal{K}|$ and select a set $S\subseteq I_K$ of size exactly $|\pal{K}|$ (\cref{lem:free-color}).
% Note that $\avganti_K \le \epsilon\Delta^2$; hence $|S| \ge \Omega(|K|)$.

Partition $I_K$ into two arbitrary disjoint sets $S_1$ and $S_2$ of size at least $\floor{|I_K|/2}$.
Order nodes $v_1, \ldots, v_{|S_1|}$ of $S_1$ by increasing \IDs.
Let $X_i$ be the random variable indicating if $v_i$ fails to get colored.
The only way a node of $S_1$ fails to adopt its color is if it conflicts with an external neighbor.
For any conditioning of values $\pi(v_j)$ for all $j < i \le |I_K|/2$, the probability $v_i$ fails to get colored is $\Pr(X_i=1~|~\pi(v_1), \ldots, \pi(v_{i-1}) ) \le \frac{e_v}{|I_K|/2} \le O\parens*{\frac{e_v}{|K|}}$, using $|I_K|/2 \ge \Omega(|K|)$.
Therefore, the expected number of uncolored nodes in $S_1$ is
\[ \Exp[X] \le \sum_{v\in S_1} O\parens*{\frac{e_v}{|K|}} \le \sum_{v\in K} O\parens*{\frac{e_v}{|K|}} = O(\avgext_K) \ . \]
By the martingale inequality (\cref{lem:chernoff}), w.h.p., at most $O(\avgext_K + \log n)$ nodes are uncolored in $S_1$.
The same reasoning applies to $S_2$ and by union bound over both sets, the number of uncolored nodes in $S$ after the synchronized color trial is at most $O(\avgext_K + \log n)$.
\end{proof}

To implement the synchronized color trial, a node $v$ needs only to know $\pi(v)$ and the $\pi(v)$-th color of $\pal{K}$.
We use an approach similar to \cite{FGHKN23}:
we randomly partition nodes into groups $T_1, \ldots, T_k$ to spread the workload.
Concretely, we use the following fact, which is a straightforward consequence of Chernoff and \cref{def:almost-clique}.\ref{part:deg-inside}.
\begin{fact}
\label{lem:random-groups}
Let $K$ be an almost-clique and $k\le |K|/(C\log n)$ for some large enough $C > 0$.
Suppose each $v\in K$ samples $t(v)\in[k]$ uniformly at random.
Then, w.h.p., each $T_i=\set{v\in K: t(v)=i}$ satisfies that any $u,w\in K$ have $|\Ntwo(u)\cap \Ntwo(w)\cap T_i|\ge (C/4)\log n$.
We say set $T_i$ \emph{$2$-hop connects} $K$.\footnote{Note that the 2-hops are in $G^2$. In $G$, this means that for any pair of nodes $u_0,u_4\in K$, there is a length-4 path $u_0,u_1,u_2,u_3,u_4$ in $G$ s.t.\ $u_2 \in T_i$.}
\end{fact}

Contrary to \cite{FGHKN23}, at distance-2, nodes do not have the bandwidth to learn the whole clique-palette nor the full random permutation.
Fortunately, they only need to know their position in the permutation and the one corresponding color.
The main technical novelty in our distance-2 implementation lies in an algorithm to compute prefix-sums $\sum_{j < i} x_j$ where each random group $T_i$ holds a value $x_i$ (\cref{lem:prefix-sum}).
We first explain how to aggregate such prefix sums and then show it is enough for implementing the synchronized color trial.

\begin{lemma}[Prefix Sums]
\label{lem:prefix-sum}
Let $T_1, \ldots, T_k\subseteq K$ be disjoint sets that $2$-hop connect $K$.
% : for each $v\in K$ and $i\in[k]$ there is at least one node in $N(v)\cap N(T_i)$.
If each $T_i$ holds a $\poly\log n$-bit integer $x_i$, then there is a $O(1)$-round algorithm such that for all $i\in [k]$, each $v\in T_i$ learns $\sum_{j < i} x_j$.
\end{lemma}

\begin{proof}
Compute a BFS tree rooted at some arbitrary $w_K \in K$ and spanning $N^2(w_K)\cap K$.
% This can be done in $O(1)$ rounds by a BFS (since $K$ has diameter at most 4).
We order distance-2 neighbors of $w_K$ with the \emph{lexicographical order induced by the BFS tree}: distance-2 neighbors $u\in N_{G^2}(w_K)$ are ordered first by $\ID(v)$, where $v$ is the parent of $u$ in the BFS tree, and then by $\ID(u)$.
Call $u_1, u_2, \ldots, u_{|N_{G^2}(w_K)\cap K|}$ distance-2 neighbors of $w_K$ with respect to that ordering.% (i.e., $\ID(u_j) < \ID(u_i)$ for all $j < i$).
For each $i\in [k]$, \emph{node $u_i$ learns $x_i$}. Since $T_i$ 2-hop connects $K$, there must exist a node $r\in N(u_i)\cap N(T_i)$ which can relay $x_i$ from its neighbor in $T_i$ to $u_i$.
For each distance-1 neighbor $v_j\in N(w_K)\cap K$ of $w_K$ (i.e., depth-1 nodes in the BFS tree), let $u_{i_j}, u_{i_j+1}, \ldots, u_{i_{j+1}-1}$ be its children in the BFS tree.
Each $v_j$ can learn all values $x_{i_j}, \ldots, x_{i_{j+1}-1}$ with a broadcast.
Node $v_j$ then sends the sum $S_j \eqdef \sum_{k=i_j}^{i_{j+1}-1} x_k$ to $w_K$, which responds with
$\sum_{k < j} S_j = \sum_{k < i_j} x_k$.
For each child $u_{i_j+t}$ with $0\le t \le i_{j+1}-i_j$, the node $v_j$ communicates
\[ \sum_{k < j} S_k+\sum_{k=i_j}^{i_j+t-1} x_k = \sum_{k < i_j + t} x_k \]
to $u_{i_j+t}$, which is exactly the prefix sum it had to learn.
Each $u_i$ can then transmit its prefix sum to $T_i$ using the same path it used to learn $x_i$.
\end{proof}

% Before explaining the algorithm for \cref{lem:prefix-sum},
We now explain how we use it to implement \sct.
Since the following algorithmic ideas are very similar to the ones of \cite{FGHKN23}, we only discuss them briefly, and claim no novelty beyond \cref{lem:prefix-sum}.

\begin{lemma}[Permute]
\label{lem:perm}
There is an algorithm that samples a uniform permutation $\pi$ of $[|I_K|]$ in $O(1)$ rounds with high probability.
The $i$-th node in $I_K$ (with respect to any ordering where $v$ knows its index) learns $\pi(i)$.
\end{lemma}

\begin{proof}
Each node $v\in I_K$ picks an integer $t(i)\in[\Theta(|K|/\log n)]$ at random.
Let $T_i=\set{v\in S: t(v) = i}$.
By Chernoff bound, w.h.p., $|T_i|=O(\log n)$ and $2$-hop connects $K$ (\cref{lem:random-groups}).
In particular, each $T_i$ has hop-diameter at most 4.
Let $w_i$ be the node of minimum \ID in $T_i$. Each $T_i$ computes a spanning tree rooted at its $w_i$. This is performed in parallel for all groups, by having nodes forward the minimum ID they received from a group $T_i$ to other members of $T_i$. Note that an edge only needs to send information concerning the two groups of its endpoints. Each $T_i$ then relabels itself using small $O(\log\log n)$-bit identifiers in the range $[|T_i|]$. $w_i$ samples a permutation $\rho_i$ of $|T_i|$ and broadcasts it to $T_i$.
Since the permutation of a group needs $O(\log n)\times O(\log\log n)$ bits, after $O(\log\log n)$ rounds each $v\in T_i$ knows $\rho_i(v)$.
Then, using \cref{lem:prefix-sum}, each $v$ learns $\sum_{j < i} |T_j|$.
Finally, node $v$ sets its position to $\pi(v)=\sum_{j < i}|T_j| + \rho_i(v)$.
\end{proof}

\begin{lemma}[Free Color]
\label{lem:free-color}
Suppose each node in $v\in K$ holds an integer $i_v\in[\Delta^2+1]$. There is $O(1)$-round algorithm at the end of which each $v$ knows the $i_v$-th color of $\pal{K}$ (with respect to any globally known total order of $\pal{K}$).
Furthermore, all nodes can learn $|\pal{K}|$ in the process.
\end{lemma}

\begin{proof}
Each node $v\in K$ picks an integer $t(v)\in [\Theta(\Delta^2/\log n)]$.
Let $T_i=\set{v\in K: t(v)=i}$.
Again, w.h.p., $|T_i|=O(\log n)$ and $T_i$ $2$-hop connects $K$.
Each node broadcasts its color (if it adopted one) and its group number $t(v)$.
Let $R_i = \set{i\cdot \Theta(\log n), \ldots, (i+1)\cdot\Theta(\log n) - 1}$.
Let $S_{u,i} = R_i\cap \col(N(u)\cap K)$ be the colors from range $R_i$ used by neighbors of $u$.
For each $i\in[k]$, node $u$ can describe $S_{u,i}$ to each neighbor in $T_i$ using a $O(\log n)$-bitmap.
Since each $T_i$ has diameter $4$ and $2$-hop connects $K$, after $O(1)$ rounds of aggregation on bitmaps using a bitwise OR, each node in $T_i$ knows $R_i\cap \col(K)$, i.e., all colors from range $R_i$ used in $K$.
Note that this also allows them to compute $R_i\setminus\col(K) = R_i\cap \pal{K}$, i.e., the colors of $R_i$ that are \emph{not} used by a node of $K$.
By \cref{lem:prefix-sum}, nodes of $T_i$ learn $\sum_{j< i}|R_j\cap\pal{K}|$ in $O(1)$ rounds.
Finally each $v$ broadcasts $i_v$ and each $u\in T_i$ broadcasts $i$, $\sum_{j < i}|R_j\cap\pal{K}|$ and $R_i\setminus \col(K)$.
Since each set $T_i$ 2-hop connects $K$, if the $i_v$-th color of $\pal{K}$ belongs to range $R_i$ (i.e., $\sum_{j < i}|R_j\cap\pal{K}| \le i_v < \sum_{j \le i} |R_j\cap\pal{K}|$), then there exists a $u\in T_i$ and $w\in N(u)\cap N(v)$ which knows both $i_v$ and the color it corresponds to.
Then $w$ can transmit that information to $v$.

To learn $|\pal{K}|$, nodes aggregate the sum of all $|R_i\cap\pal{K}|$.
This can easily done with a BFS (and electing a leader in each group to avoid double counting).
\end{proof}

\subsection{Slack Color (with extra bandwidth)}
\label{sec:slack-color-log2n}

After the synchronized color trial, uncolored nodes have degree %$O(\slack_v)$, which is
proportional to the slack they received from \slackgeneration (\cref{lem:slack-generation}).
Contrary to \cite{CLP20,HKMT21,HKNT22}, nodes cannot trivially try colors from their palettes, for they lack direct knowledge of it.
In this section, we give a solution that uses $O(\log^2 n)$ bandwidth and defer the \congest implementation to \cref{sec:slack-color-logn}.
The idea is to sample $\Theta(\log n)$ colors from the clique palette, which is accessible by \cref{lem:free-color}.
Note that this step is needed only in high-sparsity cliques: if $\avganti_K+\avgext_K \le O(\log n)$, then its remaining uncolored nodes after the synchronized color trial have degree $O(\log n)$.
This motivates the following definition:
% Therefore, we assume $\slack_K \ge \Omega(\log n)$.
%
% The following lemma states that the clique palette preserve the slack.

\begin{definition}[$\Kmod$, $\Kvery$]
\label{def:classification-acd}
Let $C>0$ be a large enough constant.
We say that almost-clique $K$ is
\emph{very dense} if $\avganti_K < C\log n$ and $\avgext_K, \errext_K < 4C\log n$.
Reciprocally, we say $K$ is \emph{moderately dense} if it is not very dense.
We call $\Kmod$ the set of moderately dense almost-cliques and $\Kvery$ the very dense ones.
\end{definition}

\cref{lem:slack} shows that in moderately dense almost-cliques, the clique palette preserves the slack provided by early steps of the algorithm (slack generation, colorful matching and degree slack).

\begin{lemma}[The Clique Palette Preserves Slack]
\label{lem:slack}
After \slackgeneration and \matching,
for all inlier $v\in I_K$ with $K\in \Kmod$,
we have \[ |\pal{v}\cap\pal{K}| \ge \hatd(v) + \Omega(\te_v+\avgext_K + \avganti_K)\ . \]
In particular, for any such $v\in I_K$ with $K\in\Kmod$, we have $|\pal{v}\cap\pal{K}| \ge \Omega(|\pal{K}|)$.
\end{lemma}

\begin{proof}
% Observe that $\Delta^2 \ge \td(v) \ge |\Ntwo(v)\cap K| + e_v + (a_v - \ta_v)$ and $|K| \le |\Ntwo(v)\cap K| + a_v$.
Clearly, $|K|= |\Ntwo(v)\cap K| + a_v$.
We carefully add all contributions to the degree slack of a node
\[ \Delta^2 = (\Delta^2 - \td(v)) + \td(v) = |\Ntwo(v) \cap K| + e_v + (\Delta^2 - \td(v)) + \errext_v + \erranti_v \ . \]
The clique palette loses one color for each colored node but saves one for each edge in the colorful matching. Recall that $\hK$ denotes the uncolored part of $K$. The clique palette has size at least
\[ |\pal{K}| \ge \Delta^2 - (|K|-|\hK|) + |M| \ge |\hK| + e_v + |M| - a_v + (\Delta^2 - \td(v)) + \errext_v + \erranti_v \ . \]
Let $s$ be the slack received w.h.p.\ by $v$ after \slackgeneration: if $e_v \ge C\log n$ then $s\eqdef\Omega(\slack_v) \ge \Omega(e_v)$ (\cref{lem:slack-generation} and \cref{lem:bound-ext}), otherwise $s=0$.
The palette of $v$ is of size at least
\[ |\pal{v}| \ge \hatd(v) + s + (\Delta^2 - \td(v)) + \errext_v + \erranti_v \ . \]
Notice $|\pal{v}\setminus\pal{K}| \le a_v$ and $|\pal{K}\setminus\pal{v}| \le \colored{e}_v$ (recall $\colored{e}_v$ is the number of \emph{colored} external neighbors).
A double counting argument bounds the number of colors in both $v$'s palette and the clique palette:
\begin{align*}
2|\pal{v}\cap\pal{K}| &= |\pal{v}| + |\pal{K}| - |\pal{v}\setminus\pal{K}| - |\pal{K}\setminus\pal{v}|\\
&\ge \hatd(v)+ |\hK| + (e_v - \colored{e}_v) + s + |M| - 2a_v + 2\erranti_v + 2\errext_v + 2(\Delta^2 - \td(v)) \\
&\ge 2\hatd(v) + s + |M| - 2\ta_v + 2\errext_v + 2(\Delta^2 - \td(v)) \ , \addtocounter{equation}{1}\tag{\theequation} \label{eq:double-counting}
\end{align*}
where the second inequality uses $|\hK| + (e_v - \colored{e}_v) \ge \hatd(v)$ and $\erranti_v - a_v = (a_v - \ta_v) - a_v = -\ta_v$.

The remaining of this proof is a careful case analysis to show that \cref{eq:double-counting} implies the result. Slack implicitly refers to the slack in the clique palette, i.e., node $v$ has slack $x$ if $|\pal{v}\cap\pal{K}| \ge \hatd(v)+x$. Each case of our analysis corresponds to a regime for $\avganti_K$ and $\avgext_K$, since $v$ receives slack from the coloring matching only when $\avganti_K > \Omega(\log n)$ (\cref{lem:colorful-matching}) and from slack generation when $e_v \ge \Omega(\log n)$ (\cref{lem:slack-generation}).
When both quantities are too small, the following fact implies nodes must have slack from a low degree.

\begin{fact}
\label{fact}
% \begin{itemi}
If $\avganti_K, \te_v \le \avgext_K/4$, then $\Delta^2-\td(v) > \avgext_K/2$.
% \item If $\avganti_K, \te_v \le \errext_K/4$, then $\Delta^2-\td(v) > \avgext_K/2$.
\end{fact}

\begin{proof}
For all $v\in K$, we have $|K|=|\Ntwo(v)\cap K|+a_v$ and $\Delta^2 = (\Delta^2-\td(v)) + |\Ntwo(v)\cap K| + \err_v + e_v$, \cref{eq:D-minus-K} holds:
\begin{equation}
\label{eq:D-minus-K}
    \Delta^2 - |K| = (\Delta^2-\td(v)) + \err_v + e_v - a_v = (\Delta^2-\td(v)) + \te_v - \ta_v \ .
\end{equation}
Since this holds for all nodes, it also holds on average:
\begin{equation}
\label{eq:D-minus-K-avg}
    \Delta^2 - |K| \ge \err_K + \avgext_K - \avganti_K \ .
\end{equation}
We conclude by replacing \cref{eq:D-minus-K-avg} in \cref{eq:D-minus-K}:
\begin{align*}
\Delta^2 - \td(v) &\ge (\Delta^2-|K|) - \te_v \tag{by \cref{eq:D-minus-K}}\\
    &\ge \avgext_K - \avganti_K - \te_v  \tag{by \cref{eq:D-minus-K-avg}}\\
    &\ge \avgext_K/2 \tag{because $\avganti_K, \te_v \le \avgext_K/4$} \ .
\end{align*}
\end{proof}

\paragraph{Case 1:} If $\avganti_K > C\log n$ and $\avgext_K < 4C\log n$\textbf{.}
We compute a colorful matching of size $|M|\ge 402\avganti_K$. Thus, all nodes have slack $|M|-2\ta_v \ge 2\avganti_K$, because $\ta_v \le 200\avganti_K$ for all inliers (\cref{lem:compute-outliers}).
If $e_v \ge \avganti_K > C\log n$, then $v$ receives slack $\Omega(e_v)$ from slack generation; hence it has $\Omega(\avganti_K + \avgext_K + \te_v)$ slack by \cref{eq:double-counting}.
Otherwise, if $\avganti_K> e_v$, it gets enough slack from the colorful matching.

\paragraph{Case 2:} If $\avganti_K > C\log n$ and $\avgext_K \ge 4C\log n$\textbf{.}
Similarly to case 1, nodes have slack $\avganti_K$.
% And nodes with $e_v \ge \avganti_K$ received enough slack from slack generation. So we assume $e_v \le $
If $\avganti_K > \avgext_K/4$ or $\errext_v \ge \avgext_K/8$, then it has enough slack.
Finally, if $e_v > \avgext_K/8 > \Omega(\log n)$, then $v$ received $\Omega(e_v)$ slack from \slackgeneration; hence has slack $\Omega(\avganti_K + \avgext_K + \te_v)$.
The only remaining possibility is that $\avganti_K, \te_v \le \avgext_K/4$. Then, \cref{fact} shows that $\Delta^2-\td(v) \ge \avgext_K/2 \ge \Omega(\avganti_K + \avgext_K + \te_v)$ and we are done.

\paragraph{Case 3:} If $\avganti_K < C\log n$ and $\avgext_K > 4C\log n$\textbf{.}
If $e_v > \avgext_K/8 \ge \Omega(\log n)$, then $v$ has slack $\errext_v + \Omega(e_v) \ge \Omega(\avganti_K + \avgext_K + \te_v)$ from slack generation, so we are done.
If $\errext_K > \avgext_K/8$, then again we are done.
Otherwise, $\avganti_K, \te_v \le \avgext_K/4$ and by \cref{fact} we conclude that all nodes have enough degree slack.

\paragraph{Case 4:} If $\avganti_K < C\log n$ and $\avgext_K < 4C\log n$\textbf{.}
Since $K\in\Kmod$, it must be that $\errext_K > 4C\log n$.
If $\te_v$ is greater than $\errext_K/8$, then $v$ has slack $\Omega(\te_v)\ge \Omega(\te_v + \errext_K)\ge\Omega(\avganti_K+\avgext_K+\te_v)$ and we are done.
So we can assume $\avganti_K, \te_v < \errext_K/4$.
We argue that the degree slack must be large. Similarly to \cref{fact}, we have
\begin{align*}
\Delta^2 - \td(v) &\ge (\Delta^2-|K|) - \te_v \tag{by \cref{eq:D-minus-K}}\\
    &\ge \errext_K - \avganti_K - \te_v \tag{by \cref{eq:D-minus-K-avg}}\\
    &\ge \errext_K/2 \ge \Omega(\avganti_K + \avgext_K + \te_v) \tag{by assumption}\ .
\end{align*}

\paragraph{Constant density.}
Observe that if $|\pal{K}| > 2e_v$ then $|\pal{v}\cap\pal{K}| \ge |\pal{K}|-e_v \ge |\pal{K}|/2$.
Otherwise if $e_v \ge |\pal{K}|/2$, we use $|\pal{v}\cap\pal{K}|\ge\Omega(e_v)$ (which we just proved) to deduce that $|\pal{v}\cap\pal{K}|\ge \Omega(e_v) \ge \Omega(|\pal{K}|)$.
\end{proof}

\cref{lem:sample-log2n} is the main implication of \cref{lem:slack}. It states that we can use random sampling in the clique palette, instead of nodes' palettes, to try colors in \slicecolor (\cref{lem:slack-color}).
In particular, after \sct (Step \cref{step:sct} in \cref{alg:high-level}), \slicecolor with the sampling process described in \cref{lem:sample-log2n} reduces degrees to $O(\log n)$ in $O(\log\log n)$ rounds.

\begin{lemma}
\label{lem:sample-log2n}
There is an $O(1)$-round algorithm (using $O(\log^2 n)$ bandwidth) that when run after \slackgeneration and \matching, achieves the following:
It samples a random color $\rC_v\in \pal{v}\cup\set{\bot}$ for all uncolored dense nodes $v\in K\in \Kmod$ such that $\Pr(\rC_v=\bot)\le 1/\poly(n)$ and $\Pr(\rC_v=c) \le \frac{1}{\hatd(v)+\Omega(\avganti_K+\avgext_K+\te_v)}$ for all colors $c\in \pal{v}\cap\pal{K}$.
\end{lemma}

\begin{proof}
Fix a node $v\in K$.
Nodes of $K$ can learn $|\pal{K}|$ by \cref{lem:free-color}.
Then $v$ samples $x=\Theta(\log n)$ indices in $[|\pal{K}|]$.
By \cref{lem:free-color}, using $O(\log^2 n)$ bandwidth, each node can learn in $O(1)$ rounds the colors corresponding to the indices they sampled.
They broadcast this list of colors (using $O(\log^2 n)$ bandwidth) and drop all colors used by neighbors (i.e., that are not in their palette).
Finally, node $v$ picks $\rC_v$ uniformly at random among the remaining ones.
Since $|\pal{K}\cap\pal{v}| \ge \Omega(|\pal{K}|)$, by sampling $\Theta(\log n)$ colors, we sample at least one color $\pal{K}\cap\pal{v}$ with high probability (i.e., $\Pr(\rC_v=\bot)\le1/\poly(n)$).
To argue about the uniformity (\cref{eq:uniform-slack-color}), we observe that sampling $x=\Theta(\log n)$ indices in $[|\pal{K}|]$ and then trying a random one of those is equivalent to sampling a uniform permutation $\pi$ of $[|\pal{K}|]$ (the $x$ sampled indices are $\pi^{-1}(1), \ldots, \pi^{-1}(x)$) and trying the color $c\in\pal{K}\cap\pal{v}$ with the smallest $\pi(c)$ (if $\min\pi(\pal{K}\cap\pal{v}) < x$).
Hence, if we call $Z=\min\pi(\pal{K}\cap\pal{v})$, we have
\begin{align*}
    \Pr(\rC_v=c) &= \Pr(Z < x~\wedge~\pi(c) = Z)\\
    % &= \Pr(\pi(c) = Z)\times\Pr(Z \le x~|~\pi(c) = Z)\\
    &\le \Pr(\pi(c) = Z) = \frac{1}{|\pal{K}\cap\pal{v}|} \\
    &\le \frac{1}{\hatd(v) + \Omega(\avganti_K+\avgext_K+\te_v)} \ . \tag{by \cref{lem:slack}}
\end{align*}
\end{proof}

\subsection{Learning Small Palettes (with extra bandwidth)}
\label{sec:learn-palette-log2n}

Assume we are given sets $L_1, \ldots, L_{\ell}$ for some $\ell=O(\log\log n)$ such that the maximum uncolored degree in each $G[L_i]$ is at most $O(\log n)$.
We explain how nodes learn a list $L(v)$ of $\hatd(v)+1$ colors in their palette, with respect to the current coloring of $G^2$.
% To complete the coloring, we go through sets $V_i$ for $i=1, \ldots, k$ sequentially, use \cref{lem:learn-palette} so that nodes know their palettes and color $V_i$ using the small degree algorithm (\cref{lem:small-degree}).

\begin{lemma}[Learn Palette]
\label{lem:learn-palette}
Let $H$ be an induced subgraph of $G^2$ with maximum uncolored degree $O(\log n)$.
There is a $O(\log\log n)$-round algorithm at the end of which each node in $H$ knows a set $L(v)\subseteq \pal{v}$ of $\hatd_H(v)+1$ colors with high probability.
\end{lemma}

The argument is two-fold, we deal with $v\in K\in \Kmod$ nodes and very dense nodes $v\in K\in\Kvery$ separately.
% The former (thanks to the colorful matching) have a large enough palette for a simple sampling algorithm to work.
% For the latter nodes, we use the density of their almost-cliques to learn colors used by anti-neighbors and compensate for the absence of a colorful matching.
Here, we assume $O(\log^2 n)$ bandwidth.
The $O(\log n)$ bandwidth argument can be found in \cref{sec:learn-palette-logn}.

\paragraph{Moderately Dense Almost-Cliques.}
Using the sampling algorithm from \cref{lem:sample-log2n}, nodes can sample $C\log n$ many colors in their palette in $O(1)$ rounds, for any arbitrarily large constant $C > 0$. Since uncolored degrees in $H$ are $O(\log n)$, this suffices for \cref{lem:learn-palette}.

\begin{lemma}
\label{lem:learn-moderatly-dense-log2n}
Let $H$ be an induced subgraph of $G^2$ with maximum uncolored degree $C'\log n$ for a large constant $C' > 0$.
% Let $v\in K\in\Kmod$ and $\hatd(v)< C'\log n$ for a large enough constant $C' > 0$.
There is a $O(1)$-round algorithm (using $O(\log^2n)$ bandwidth) for $v\in K\in\Kmod$ to learn a list $L(v)$ of $\hatd_H(v)+1$ colors from their palettes.
\end{lemma}

\begin{proof}
Let $\gamma > 0$ be the universal constant such that $|\pal{v}\cap\pal{K}| \ge \gamma|\pal{K}|$ from \cref{lem:slack} (which applies because $v\in K\in\Kmod$).
Each $v$ builds a set $S_v$ by sampling each color $c\in \pal{K}$ in it with probability $p\eqdef \frac{2C'\log n}{\gamma|\pal{K}|}$.
In expectation, $v$ samples $\Exp[|S_v\cap\pal{v}|] = p|\pal{K}\cap\pal{v}| = 2C'\log n \times \frac{|\pal{K}\cap\pal{v}|}{\gamma|\pal{K}|} \ge 2C'\log n$ colors. By Chernoff bound, w.p.\ $1-n^{-\Theta(C')}$, $|S_v\cap\pal{v}| \ge C'\log n \ge \hatd_H(v)+1$.
To implement the sampling of colors, each node samples indices in $[|\pal{K}|]$ and learns the corresponding colors with \cref{lem:free-color}.
Since $|S_v|=O(\log n)$, a node can broadcast this set to its distance-1 neighbors and remove colors used by external neighbors using $O(\log^2 n)$ bandwidth.
Hence, each $v$ finds a list $L(v) \eqdef S_v\cap\pal{v}$ of at least $\hatd_H(v)+1$ colors.
\end{proof}

\paragraph{Very Dense Almost-Cliques.}
We are now considering nodes $v\in K\in\Kvery$, such that \cref{lem:slack} does not apply.
We use the following broadcast primitive: Each node forwards along each outgoing edge a (independently) random message received.

\begin{lemma}[Distance-2 Many-to-All Broadcast]
\label{lem:many-to-all}
Let $K$ be an almost-clique in $G^2$ and $S\subseteq K$ be a subset of $k$ vertices such that each $x\in S$ has a message $m_x$. Suppose $\Delta \geq k \log n$ and $\Delta^2 \geq k^3\log n$. After four rounds of the broadcast primitive, every node in $K$ received all messages $\set{m_x}_{x\in S}$, w.h.p.
%If every node broadcasts a random message it receives for 3 rounds, w.h.p., every node in $K$ receives in one round all messages $\{m_x\}_{x\in S}$.
\end{lemma}
\begin{proof}
Let $u \in S$ and $v \in K$. Recall that $u$ and $v$ both have at least $(1-\eps)\Delta^2$ d2-neighbors in $K$. Since $\card{K} \leq (1+\eps)\Delta^2$, there are at least $(1-3\eps)\Delta^2$ common d2-neighbors of $u$ and $v$ in $K$. Let $W \eqdef \Ntwo(u) \cap \Ntwo(v)$ be this set.

We attribute a unique ``relay'' to each node $w\in W$, connecting it to $v$. For each $w\in W$, let $r_w$ be the common d1-neighbor of $w$ and $v$ of lowest \ID. For each d1-neighbor $r$ of $v$, let $W_r \subseteq W$ be the nodes of $W$ for which $r$ is the chosen relay to $v$. Assume $\eps < 1/12$. Using that $\card{W_r} \leq \Delta$, and by a simple Markov-type argument, there are at least $(1-6\eps)\Delta \geq \Delta/2$ d1-neighbors $r$ of $v$ for which $\card{W_r} \geq \Delta/2$. Let $R$ be the set of those ``heavy'' relays.

After the first round, each d1-neighbor of $u$ receives $m_u$.
Consider some heavy relay $r\in R$. Each node $w \in W_r$ receives the message $m_u$ from a d1-neighbor it shares with $u$ with probability at least $1/k$, independently from other nodes. Thus, with probability at least $1-\exp(-\Delta/(24k)) = 1-1/\poly(n)$, $\Delta/(4k)$ or more nodes in $W_r$ receive $m_u$.

Assume at least $\Delta/(4k)$ nodes in $W_r$ received $m_u$. Then, in the third round of the broadcast primitive, $r$ fails to receive $m_u$ with probability at most:
\[
\parens*{1-\frac 1 k}^{\Delta/(4k)} \leq e^{-\Delta/(4k^2)}\ .
\]
When $\Delta / (4k^2) \geq 1/2$, this probability is bounded by $e^{-1/2} \leq 2/3$. In that case, $\Delta/6$ or more nodes in $R$ should receive $m_u$ in expectation, and so at least $\Delta/12$ heavy  relays receive $m_u$ w.p.\ $1-\exp(-\Delta/72)$. The probability that those relays all fail in sending $m_u$ to $v$ in the fourth round of the broadcast primitive is at most $(1-1/k)^{\Delta/72} \leq \exp(-\Delta/(72k))=1/\poly(n)$.

If $\Delta / (4k^2) \leq 1/2$, then $e^{-\Delta/(4k^2)} \leq 1-\Delta/(8k^2)$. In expectation, at least $\Delta^2/(16k^2)$ heavy relays receive $m_u$, and $\Delta^2/(32k^2)$ of them receive $m_u$ w.h.p. All those relays fail to send $m_u$ to $v$ with probability at most $(1-1/k)^{\Delta^2/(32k^2)} \leq \exp(-\Delta^2/(32k^3)) = 1/\poly(n)$.
\end{proof}

% In particular, for distance-2 coloring with $\Delta^2+1$ colors, all nodes receive the message as long as $\Delta \ge 4 k^2 \log n$.
% This allows us, in \cref{lem:learn-palette-very-dense-log2n}, to correct for errors in the clique palette approximation.

In particular, this allows us to learn all colors remaining in the clique palette, because at this step of the algorithm, only $\poly\log n$ colors should remain available in $\pal{K}$. If not, we learn nonetheless a set of $\poly\log n$ colors from $\pal{K}$ which will act as a replacement.

\begin{lemma}
\label{lem:learn-colors-clique-palette}
Assume $\Delta \ge \Omega(\log^{3.5} n)$.
There is an $O(1)$-round algorithm (with $O(\log n)$ bandwidth) such that, in each almost-clique $K$, either
\begin{enumerate}
    \item all nodes $v\in K$ learn all colors in $\pal{K}$, or
    \item all nodes learn a set $D\subseteq \pal{K}$ of $\Theta(\log^2 n)$ colors.
\end{enumerate}
\end{lemma}

\begin{proof}
If $|\pal{K}| \le O(\log^2 n)$, then let $D\eqdef\pal{K}$. Otherwise, if $|\pal{K}| \ge \Omega(\log^2 n)$, let $D$ be the $\Theta(\log^2 n)$ first colors of $\pal{K}$.
Recall that all nodes can learn $|\pal{K}|$ in $O(1)$ rounds (\cref{lem:prefix-sum}); hence, nodes know in which of the two case they are.

Assign indices of $[|\pal{K}|]$ to arbitrary nodes $u_1, \ldots, u_{|D|}$ of $K$ (with a BFS for instance).
This is feasible because $|K| \ge \Delta^2/2 > \Theta(\log^2 n) = |D|$.
Then, each $u_i$ learns the $i$-th color of $D$ in $O(1)$ rounds (\cref{lem:free-color}).
Each $u_i$ then crafts a message $m_i$ containing that color and distributes it to all nodes of $K$ by Many-to-All broadcast.
By assumption, there are only $|D|=O(\log^2 n)$ messages, and since $|K| \ge \Delta^2/2 \ge \Theta(\log^{7} n)=|D|^3\times \Theta(\log n)$, we meet the requirements of \cref{lem:many-to-all}. Thus, in $O(1)$ rounds, all the nodes in $K$ know all colors of $D$.
\end{proof}

For nodes of very dense almost-cliques, the clique palette $\pal{K}$ does not approximate their palette well enough. We correct for that by adding colors used by anti-neighbors. They filter out colors used by external neighbors with $O(\log^2 n)$ bandwidth, because they have $O(\log n)$ such neighbors.

\begin{lemma}
\label{lem:learn-palette-very-dense-log2n}
% Let $H$ be an induced subgraph of $G^2$ with maximum uncolored degree $O(\log n)$.
Suppose each $K\in\Kvery$ has $O(\log n)$ uncolored nodes (hence $\hatd(v) \le O(\log n)$ for all $v\in K\in\Kvery$).
There is a $O(\log\log n)$-round randomized algorithm (using $O(\log^2 n)$ bandwidth) for all uncolored nodes $v\in K\in\Kvery$ to learn a list $L(v)$ of $\hatd_H(v)+1$ colors from their palettes.
\end{lemma}

\begin{proof}
Run the algorithm of \cref{lem:learn-colors-clique-palette}. Assume first that nodes learned all colors of $\pal{K}$.
Recall nodes have $e_v = O(\log n)$ external neighbors (because $K\in\Kvery$ and $v\in I_K$); hence, they can learn all colors used by their external neighbors in $O(1)$ rounds by using $O(\log^2 n)$ bandwidth.
Since each uncolored $v$ knows $\pal{K}$ and the colors of its external neighbors, it thereby knows $\pal{K}\cap\pal{v}$.

With a BFS, we can relabel uncolored node of $K$ in the range $[O(\log n)]$.
Since uncolored nodes are inliers, they have anti-degree $a_v \le O(\avganti_K) \le O(\log n)$ each.
At most $O(\log^2 n)$ nodes in $K$ are anti-neighbors of (at least one) uncolored node.
We can run $O(\log n)$ BFS in parallel, one rooted at each uncolored node, such that each node knows to which uncolored node it is connected (at distance-2).
This takes $O(\log\log n)$ rounds, even with bandwidth $O(\log n)$, because each BFS uses $O(\log\log n)$-bit messages (thanks to the relabeling) and we run $O(\log n)$ of them.
Then, the $O(\log^2 n)$ nodes with an uncolored anti-neighbor can describe their list of uncolored anti-neighbors using a $O(\log n)$-bitmap.
Using \cref{lem:many-to-all}, they broadcast this information as well as their color to all nodes.

Nodes use lists $L(v) \eqdef (\pal{K}\cup\col(K\setminus \Ntwo(v)))\cap\pal{v}$, i.e., the clique palette augmented with the colors of their anti-neighbors, minus colors used by external neighbors.
Adding $\Delta^2 + 1 \ge |\Ntwo(v)\cap K| + e_v$ and $|K| \le |\Ntwo(v)\cap K| + a_v$, we get $|\pal{K}| \ge \Delta^2 + 1 - (|K|-|\hK|) \ge |\hK|+1 +e_v- a_v$. Since each colored external neighbor removes as most one color, lists have size (recall $\colored{e}_v$ and $\colored{a}_v$ are the \emph{colored} external degree and anti-degrees respectively)
\[ |L(v)| \ge |\pal{K}|-\colored{e}_v + \colored{a}_v \ge |\hK| + (e_v-\colored{e}_v) - (a_v - \colored{a}_v) + 1 = \hatd(v)+1 \ . \]

Suppose now that we are in the second case of \cref{lem:learn-colors-clique-palette}, i.e., nodes learn a set $D\subseteq \pal{K}$ of $\Theta(\log^2 n)$ colors. This immediately leads to a good approximation.
% the assumption that $|\pal{K}| \ge \Omega(\log^2 n)$ fails, then restricting ourselves to the $O(\log^2n)$ smallest colors of $\pal{K}$ already gives large lists for all nodes.
% Call this set of colors $D \subseteq \pal{K}$.
Since $K\in\Kvery$, the average node has few external connections, $\avgext_K+\errext_K = O(\log n)$ (\cref{def:classification-acd}). Moreover, because uncolored nodes are all inliers, $e_v = O(\avgext_K + \errext_K)$ (\cref{eq:inliers}). Finally, node $v$ loses at most one color in $D$ per external neighbor, hence
\[
|D \cap\pal{v}| \ge |D| - e_v \ge \Theta(\log^2 n) - O(\avgext_K + \errext_K) \ge \Theta(\log^2 n) \ge \hatd(v)+1 \ .
\]
The last inequality holds because, at this point of the algorithm, nodes have $\hatd(v) = O(\log n)$.
% Clearly, the prefix-sum algorithm (\cref{lem:prefix-sum})  works for $D$ instead of $|\pal{K}|$ (because we choose $D$ to be the \emph{smallest} colors in $\pal{K}$) and nodes can query for the $i$-th color of $D$ using \cref{lem:free-color}. Once all nodes know $D$, they learn colors of their external neighbors in $O(1)$ rounds with $O(\log^2n)$ bandwidth and compute locally $L(v)=D\cap\pal{v}$.
\end{proof}

\subsection{Proof of
\texorpdfstring{\cref{thm:d2}}%
{Theorem \ref{thm:d2}}}

Let $C > 0$ be some large universal constant.
By \cref{lem:acd}, computing the almost-clique decomposition $\Vsparse, K_1, \ldots, K_k$ of $G^2$ takes $O(1)$ rounds of \congest (Step~\ref{step:acd}). Generating slack (\cref{alg:slackgeneration}) and coloring $\Vsparse$ takes $O(\log^* n)$ rounds (\cref{lem:slack-generation,lem:coloring-sparse}). Putting together all results from \cref{sec:dense}, we prove the proposition stated earlier, which implies \cref{thm:d2}.

\propColoringDenseNodes*

\emph{(Steps~\ref{step:matching}, \ref{step:outliers} \& \ref{step:color-outliers}.)}
By \cref{lem:colorful-matching}, we compute a colorful matching of size $402\avganti_K$ in $O(1)$ rounds, in all almost-cliques with $\avganti_K > C\log n$. By \cref{lem:compute-outliers}, we can compute sets $O_K$ and $I_K=K\setminus O_K$ in all almost cliques, such that all $v\in I_K$ verify \cref{eq:inliers} and $|I_K| > (1-5/100)|K|$. Let $H_1$ be the subgraph of $G^2$ induced by $\bigcup_{K} O_K$. Note that each $v\in O_K$, for some almost-clique $K$, has at least $(1-5/100)|K| -\epsilon\Delta^2 > (1-5/100)(1-\epsilon)\Delta^2 - \epsilon\Delta^2 \ge \Delta^2/2$ neighbors in $I_K$, for $\epsilon$ small enough. Hence, each outlier has $\Delta^2/2$ slack in $H_1$ and can thus be colored in $O(\log^* n)$ rounds by \cref{lem:coloring-sparse}.

\emph{(Step~\ref{step:sct} \& \ref{step:slice-color}.)}
Order nodes of $I_K$ with a BFS. By \cref{lem:perm}, w.h.p., the $i$-th node of $I_K$ can learn $\pi(i)$, where $\pi$ is a uniformly random permutation of $[|I_K|]$. By \cref{lem:free-color}, each node can learn, thus try, the $i$-th color of $\pal{K}$ (if it exists). With high probability, by \cref{lem:sct}, each almost-clique $K$ has $O(\avganti_K+\avgext_K+\log n)$ uncolored nodes. This implies, the uncolored degree of a dense node $v\in K$ is $O(e_v + \avganti_K+\avgext_K+\log n)$. In particular, if $v\in K\in \Kvery$ (\cref{def:classification-acd}), it has uncolored degree $\hatd(v) \le O(\log n)$.
Since Step~\ref{step:slice-color} intend to reduce the uncolored degree to $O(\log n)$, we can focus on moderately dense almost-cliques. Let $H_2=\bigcup_{K\in\Kmod} K$.
The following fact shows conditions of \cref{lem:slack-color} are verified by the sampler of \cref{lem:sample-log2n}.
\begin{fact}\label{fact:bound-uncolored-degree}
    There exists a universal constant $\alpha > 0$ such that $s(v) \ge \alpha\cdot b(v)$ for all $v\in H_2$, where $b(v) = \te_v + \card{\hK}$ and $s(v) \ge \Omega(\avganti_K + \avgext_K + \te_v)$ from \cref{lem:sample-log2n} such that $\Pr(\rC_v = c) \le \frac{1}{\hatd(v)+s(v)}$.
    % can be efficiently computed by all dense nodes $v$, and satisfies the hypotheses of \cref{lem:slack-color}.
\end{fact}
\begin{proof}\renewcommand{\qed}{}
    Let $K$ be the almost-clique of $v$.
    The quantity $b(v)$ only requires $O(1)$ rounds to compute: To compute its pseudo-external degree $\te_v$, a node only needs to receive from each of its direct neighbors $u \in N_G(v)$ the value $\card{(N_G(u) \cup\set{u})\setminus K}$; for $\card{\hK}$, the number of uncolored nodes in $K$, a simple BFS within $K$ suffices to count $\card{\hK}$ and broadcast it to the whole almost-clique.

    We now show $b(v)$ satisfies the hypotheses. After SCT, by \cref{lem:sct}, at most $O(\avgext_K + \avganti_K + \log n)$ nodes are left uncolored in $K$, so $b(v) \in O(\te_v+\avgext_K + \avganti_K + \log n)$. By \cref{lem:sample-log2n}, there exist $s(v) \in \Omega(\avganti_K+\avgext_K+\te_v)$ s.t.\ \cref{eq:uniform-slack-color} holds.
    If $b(v)\le C\log n$, then $\hatd(v) < C\log n$, hence the uncolored degree is already $O(\log n)$.
    Otherwise, when $b(v) \geq C \log n$, it must be that $b(v) \in \Theta(\te_v+\avgext_K + \avganti_K)$, and so, there exists a universal constant $\alpha$ s.t.\ $s(v) \geq \alpha \cdot b(v)$. \Qed{fact:bound-uncolored-degree}
\end{proof}
Hence, we can use the sampling algorithm of \cref{lem:sample-log2n} (and \cref{lem:learn-moderately-logn} with $O(\log n)$ bandwidth), to run \slicecolor (\cref{lem:slack-color}) in $H_2$.
Therefore, in $O(\log\log n)$ rounds, we produce a coloring and a partition $L_1, \ldots, L_{\ell}$ of uncolored nodes in $H_2$ such that the maximum uncolored degree of each $G[L_i]$ for $i\in [\ell]$ is $O(\log n)$.
We also define $L_0=\bigcup_{K\in\Kvery} K$ which has maximum uncolored degree $O(\log n)$ after the synchronized color trial.

\emph{(Steps~\ref{step:learn-palette} \& \ref{step:small-degree}.)}
We go through layers $L_0, L_1, \ldots, L_{\ell}$ sequentially.
In $L_0$, nodes learn lists of $\deg+1$ colors from their palette by \cref{lem:learn-palette-very-dense-log2n} (and \cref{lem:learn-palette-very-dense-logn} with $O(\log n)$ bandwidth). In each $L_i$ for $i \in [\ell]$, nodes are moderately dense, hence learn their palette from sampling by \cref{lem:learn-moderatly-dense-log2n} (and \cref{lem:learn-moderately-logn} with $O(\log n)$ bandwidth).
Solve each of these $\deg+1$-list-coloring instance of \cref{lem:coloring-dense} with the small degree algorithm of \cref{lem:small-degree}.
Since learning palettes takes $O(\log\log n)$ and each $\deg+1$-list-coloring instance is solved in $O(\log^5\log n)$, the total round complexity of this step is $O(\log^6\log n)$, which dominates the complexity of the algorithm.
\Qed{thm:d2}

\section{Reducing The Degree With Slack}
\label{sec:slice-color}

In this section we prove the following lemma:
\lemmaSlackColor*

We will actually prove a slightly stronger statement: that each node $v$ in a layer $L_i$ has $O(\log n)$ uncolored neighbors in higher or equal layers, i.e., in $L_{\geq i} \eqdef \bigcup_{j\geq i} L_j$.
The technique is similar to an algorithm of \cite{BEPS}, which partitioned uncolored nodes into $2$ subgraphs of maximum degree $O(\log n)$.

%We present an algorithm which achieves the result claimed in \cref{lem:slack-color} .
At a high level, the efficiency of \cref{alg:slicecolor} comes from a feedback loop between uncolored degree and probability of getting colored: the smaller uncolored degree is w.r.t.\ slack, the larger the probability of getting colored; the larger the probabilities that nodes get colored, the smaller the uncolored degree becomes.
We call \emph{skew} the factor $\kappa$ in \cref{eq:uniform-slack-color} by which some colors are more likely to be tried.
Intuitively, we compensate for the skew by having nodes only try colors with probability $O(1/\kappa)$ in the first phase of the algorithm, and adding a $\kappa$ term in the ratio between slack and degree for the second phase.

%As in many other places in this paper, a subtlety of \cref{alg:slicecolor} is that nodes neither know their slack nor their uncolored degree.

\begin{algorithm}[ht]
\caption{\slicecolor.}
\label{alg:slicecolor}
\Input{uncolored dense nodes with bound $b(v)\geq \hatd(v)$ on their uncolored degree, slack $s(v) \geq \alpha\cdot b(v) \in \Omega(\log n)$, given a color sampler with skew $\kappa$.}
\Output{uncolored nodes partitioned into $O(\log \log \Delta)$ graphs of $O(\log n)$ degree.}

\BlankLine

For $i=1, 2, \ldots, \ell=\ceil{\log \log \Delta}$, let $L_i \gets \set{v:b(v) < C\log n \cdot 2^{2^i}} \setminus \bigcup_{j<i} L_j $.
%Let $L_1 \gets \set{v:\pud(v) \leq O(\log n) }$, $L_i \gets \set{v \notin L_{i-1}:\pud(v) \leq 2^{2^{i-1+\floor{\log \log \log n}}} }$.

In color tries, priority is given first to higher layers, second to higher IDs

\For{$16\kappa\cdot \ln(2\kappa/\alpha)$ rounds, for each uncolored node $v$ in parallel}{%
    \label{step:sliceStartFirstLoop}%
    W.p.\ $1/(4\kappa)$, $v$ independently decides to try a color in that round.
    \label{step:sliceSilence}%

    \If{$v$ is to try a color in that round}{$v$ samples and tries a color.}\label{step:sliceEndFirstLoop}%
}

\For{$2\ceil{\log \log \Delta}+2$ rounds, for each uncolored node $v$ in parallel}{%
    \label{step:sliceSecondLoop}
    $v$ samples and tries a color.\label{step:sliceSlackColorTry}%
}
\end{algorithm}

\begin{proof}[Proof of \cref{lem:slack-color}]
For any node $v$ in a layer $L_i$, let $\uncolored{d'}(v)$ be its uncolored degree within $L_{\geq i} = \bigcup_{j\geq i} L_j$.
Throughout this analysis, the uncolored degree only counts \emph{active} neighbors of a node, i.e., uncolored nodes that are executing the \slicecolor. Notably, it ignores nodes of low uncolored degree $O(\log n)$ which are not passed to \slicecolor.

\begin{fact}[Minimum slack in layers]
    \label{claim:initSlack}
    For all $i \geq 1$, the minimum initial slack in layer $L_i$ and above, denoted $s(L_i)$, satisfies
    $s(L_i) \geq \alpha C \log n\cdot 2^{2^{i-1}}$.
\end{fact}
\begin{proof}
    By assumption, we have $s(v) \geq \alpha \cdot b(v)$, and for all nodes in layer $L_i$, $b(v) \geq C \log n \cdot 2^{2^{i-1}}$ as they would otherwise be in a previous layer.
\end{proof}

\begin{fact}[Constant degree reduction]
    \label{claim:setupRatio}
    After the first loop in \owntocref{lines}{step:sliceStartFirstLoop}{step:sliceEndFirstLoop}, w.h.p., each node $v$ has at most $s(v) / (2\kappa)$ uncolored neighbors.
\end{fact}
\begin{proof}
    Consider a node $v$, uncolored at the start of some iteration of the loop in \owncref{line}{step:sliceStartFirstLoop}. Suppose it has at most $s(v)/(2\kappa)$ (active) neighbors. Then, by union bound, if it tries a random color from the skewed color sampling algorithm, it fails to get colored with probability at most $\frac{s(v)}{2\kappa} \cdot \frac{\kappa}{\hatd(v) + s(v)} \leq 1/2$. Suppose now the opposite, that $v$ has more than  $s(v)/(2\kappa)$ neighbors. Since $s(v) \geq \Omega(\log n)$, w.h.p., less than a $1/(2\kappa)$ fraction of its neighbors decide to try a color in \owncref{line}{step:sliceSilence}. Thus, when trying a random color from the skewed color sampling algorithm, $v$ fails to get colored with probability at most $\frac{\hatd(v)}{2\kappa}\cdot \frac{\kappa}{\hatd(v) + s(v)} \leq 1/2$.

    Consider a node $v$ with $\Omega(\log n)$ neighbors. Conditioned on a high probability event, nodes try a random color w.p.\ $1/(4\kappa)$ and will each get colored with probability $1/2$ when doing so. By \cref{lem:chernoff} (Chernoff bound), at least an $1/(16\kappa)$ fraction of $v$'s neighbors get colored, w.h.p. Thus, after $16\kappa\cdot \ln(2\kappa/\alpha)$ rounds, only a $(1-1/(16\kappa))^{16\kappa\cdot \ln(2\kappa/\alpha)} \le \alpha/(2\kappa)$ fraction of its neighbors remain uncolored.
\end{proof}

\begin{fact}[Slack increases coloring probability]
    \label{claim:probaColoring}
    Suppose a node $v$ satisfies $s(v) \geq x\cdot \kappa\uncolored{d'}(v)$. When trying a color in \owncref{line}{step:sliceSlackColorTry}, it fails to get colored with probability at most $1/x$, even conditioned on arbitrary random choices by its neighbors.
\end{fact}
\begin{proof}
    The number of distinct colors tried by neighbors of $v$ in equal or higher layers is at most $\uncolored{d'}(v)$. The probability that $v$ hits one of those colors is bounded by $\uncolored{d'}(v)\cdot \kappa/(\hatd(v)+s(v)) \leq (s(v)/(x\cdot\kappa))\cdot \kappa/s(v) \leq 1/x$.
\end{proof}

\begin{fact}[Quadratic progress in ratio]
    \label{claim:quadraticProgress}
    For some $x\geq 2$, $i \geq 2$, suppose that for all $j \geq i$, all nodes $v \in L_{j}$ satisfy $\uncolored{d'}(v) \leq s(v) / (x \cdot \kappa)$.
    For each node $v\in L_i$ and any $t \geq s(v) / (x^2 \cdot \kappa)$, after $O(1)$ iterations of the loop in
    \owncref{line}{step:sliceSecondLoop}, $v$ satisfies
    $\uncolored{d'}(v) \leq t$ w.p.\ $1-\exp(-\Omega(t))$.
\end{fact}
\begin{proof}
    Let us order uncolored nodes first by layer (giving priority to layers of higher index), second by ID, such that nodes early in the order have priority over nodes later in the order when trying colors. Consider a node $v \in L_i$, and for all its neighbors in $L_{\geq i}$, consider the random variables $X_1,\ldots,X_{\uncolored{d'}(v)}$ indicating whether each of its neighbors stays uncolored after $2$ rounds of trying colors. For each such variable, by \cref{claim:probaColoring}, $\Exp[X_j] \leq 1/x^2 \leq 1/(2x)$, and so $\Exp[\sum_{j=1}^{\uncolored{d'}(v)} X_j] \leq t/2$. Note that whether a node gets colored only depends on its neighbors of higher priority, so revealing the random choices of nodes in their order of priority allows us to apply \cref{lem:chernoff} (Chernoff bound), which gives:
    \[
    \Pr\parens*{\sum_{j\in[\uncolored{d'}(v)]} \!\!\!\!X_j \geq t } \leq \exp(-t/6)\ .\qedhere
    \]%
\end{proof}

We can now prove \cref{lem:slack-color}. Initially, all nodes satisfy $s(v) \geq \alpha \hatd(v)$. By \cref{claim:setupRatio}, after the first loop (\owntocref{lines}{step:sliceStartFirstLoop}{step:sliceEndFirstLoop}), they all satisfy $s(v) \geq 2\cdot \kappa\hatd(v)$. We are now ready for repeated applications of \cref{claim:quadraticProgress}.

At the beginning of the first iteration of the second loop (\owncref{line}{step:sliceSecondLoop}), nodes satisfy $s(v) \geq 2^{2^0}\cdot \kappa\hatd(v)$. Suppose that for some integer $k$, all nodes in $L_{\geq i}$ satisfy $s(v) \geq 2^{2^k}\cdot \kappa\uncolored{d'}(v)$ and that $s(L_i) / (2^{2^{k+1}}\cdot \kappa ) \in \Omega(\log n)$. By \cref{claim:quadraticProgress}, after 2 rounds, all such nodes satisfy $s(v) \geq 2^{2^{k+1}}\cdot \kappa\uncolored{d'}(v)$, with high probability.

Since $s(L_i) / (2^{2^{i}}\cdot \kappa ) \in \Omega(\log n)$, for the $i \leq \ceil{\log \log \Delta}$ first iterations of the second loop, we can apply \cref{claim:quadraticProgress} and claim that, w.h.p., all nodes $v \in L_{\geq i}$ satisfy $s(v) \geq 2^{2^i}\cdot \kappa \uncolored{d'}(v)$. We cannot, however, reapply \cref{claim:quadraticProgress} in the same manner beyond this point, as $s(v) / (2^{2^{i+1}} \cdot \kappa)$ could now be below $c \log n$ for a small constant $c$. However, the $2^{2^i} \kappa$ ratio between $s(v)$ and $\uncolored{d'}(v)$ achieved by all nodes in $L_{\geq i}$ means that they all get colored w.p.\ at least $1-2^{-2^i}$ in subsequent iterations of the loop, by \cref{claim:probaColoring}. Since each node $v \in L_i$ has less than $C \log n \cdot 2^{2^{i+1}}$ neighbors to begin with, its expected number of $L_{\geq i}$-neighbors after $2$ more iterations of the loop is some $O(\log n)$. So, applying \cref{claim:quadraticProgress}, w.h.p., each node in $L_i$ has at most $O(\log n)$ neighbors in highers layers at the end of the algorithm.
\end{proof}

\section{Selecting Outliers}
\label{sec:outliers}

In this section we prove the following lemma:
\lemmaOutliers*

The crux is that most nodes in $K$ verify (stronger versions of) \cref{eq:inliers}.
Indeed, for anti-degrees we have access to underestimates $\ta_v \le a_v$.
Therefore, Markov inequality directly implies that at least $(1-1/100)|K|$ nodes have $\ta_v \le a_v \le 100\avganti_K$. Similarly, at least $(1-1/50)|K|$ nodes are such that $\te_v = e_v + \errext_v \le 100(\avgext_K + \errext_K)$.

% We can show something similar for pseudo external degrees.

% \begin{fact}
% \label{lem:inliers-ext-deg}
% There is $(1/2-2\epsilon-1/100)|K|$ nodes such that $e_v \le \te_v \le 12\slack_K + 100\errext_K$.
% \end{fact}

% \begin{proof}
% In \cite[Lemma 4]{HKNT22}, they show there is a set of $(1/2-2\epsilon)|K|$ such that $e_v \le 12\slack_K$.
% By Markov inequality, at most $|K|/100$ nodes have $\errext_v \ge 100\errext_K$.
% Hence, there are at least $(1/2-2\epsilon-1/100)$ nodes with $\te_v = e_v + \errext_v \le 12\slack_K + 100\errext_K$.
% \end{proof}

The issue is that, although nodes know the values of $\te_v$ and $\ta_v$, they neither know $\avgext_K$ nor $\avganti_K$.
We begin by describing a more general \emph{filtering algorithm} and then use it to select outliers.

\paragraph{A Filtering Algorithm.}
We focus on an almost-clique $K$ where each node has a value $0\le x_v \le U$ for some integer upper bound $U \le \poly(\Delta)$.
For some density parameter $\delta\in(0,1)$, our idea is to find a $\tau$ such that $\set{v\in K: x_v \ge \tau}$ has density at most $\delta$ (i.e., size $\delta|K|$).
Since we cannot afford to learn all $x_v$, we approximate their distribution and learn instead the size of the sets
\[ S_i = \set*{v\in K: (1+\eta)^i \le x_v < (1+\eta)^{i+1}}\quad\text{for } i=0,1, \ldots, \ceil{\log_{1+\eta} |K|}+1 \ , \]
where $\eta=\Theta(\delta/(1-\delta))$ is small constant parameters to be defined precisely later.

Learning all $|S_i|$ is too expensive, so we settle for a constant factor approximation.
Let $v$ be the node of minimum \ID in $K$ and $T$ be a BFS tree spanning $K$ rooted at $v$.
We shall call $T_{u}$ the subtree rooted at $u$.
% Leaves of $T$ send the index to which they belong.
Let $u$ be an internal node and $u_1, \ldots, u_d$ its children in $T$.
Suppose that each $u_j$ has a value $s_i^{u_j}$ for each $i\in[\ell]$ (where $\ell\eqdef \ceil{\log_{1+\eta} |K|}+1$) which is an approximate value of $|S_i\cap T_{u_j}|$.
Node $u$ computes the following approximation of $|S_i\cap T_u|$, where $\mathbf{1}$ is the indicator function:
\[
s^u_i = \mathbf{1}(u\in S_i) + \sum_{j=1}^d (1+\eta)^{\ceil{\log_{1+\eta}s_i^{u_j}}}
\le \mathbf{1}(u\in S_i) + (1+\eta)\sum_{j=1}^d s_i^{u_j} \ .
\]

Let $s_i=s_i^v$ be the value computed by the root.
The following fact holds because $T$ has depth 4 in $G$ and each node belongs to exactly one $S_i$.
\begin{fact}
\label{fact:si-size}
$|S_i| \le s_i \le (1+\eta)^3|S_i|$
\end{fact}

The next important fact is that aggregates $s_i^u$ are easy to compute
\begin{fact}
Node $u$ computes each $s_i^u$ for $i\in[\ell]$ in $O(\log\log\Delta)$ rounds.
\end{fact}
That is because, starting from leaves all the way to the root, each child $u_j$ of $u$ in $T$ needs to send $\ceil{\log_{1+\eta} s_i^{u_j}}$ for each $i\in[\ell]$, which is a total of $O(\log\log \Delta)\times O(\log n)$ bits to $u$.
Thus, after $O(\log\log\Delta)$ rounds, the root of the BFS tree knows all values $s_i$ for $i\in[\ell]$.
We state now the key lemma:

\begin{lemma}[Filtering Lemma]
\label{lem:filter}
Fix some $M\ge 1$ and $\delta\in(0,4/5)$.
Suppose $A=\set{v\in K: x_v \le M}$ has size at least $(1-\delta)|K|$.
There is an algorithm computing a set $\tA\subseteq K$ such that each $v\in \tA$ verifies $x_v \le 2M$ and $|\tA| \ge (1-1.5\delta)|K|$.
\end{lemma}

\begin{proof}
Define
\[\tau = \inf\set[\Bigg]{i\in[\ell]: \sum_{j\le i} s_j \ge \parens{1-\delta}|K|} \ ,\]
and let $\tA=\bigcup_{j \le \tau} S_j$.
Note that the root of $T$ knows all $s_j$'s and can therefore compute $\tau$. Once the value of $\tau$ is known, each node knows if it belongs to $\tA$. By \cref{fact:si-size} and definition of $\tau$,
\[ |\tA| = \sum_{j \le \tau} |S_j| \ge \sum_{j\le \tau} s_j/(1+\eta)^3 \ge \frac{1-\delta}{(1+\eta)^3}|K| \ge (1-1.5\delta)|K| \ , \]
where the last inequality comes from setting $\eta\eqdef \frac{\delta}{4k(1-\delta)}$ where $k=3$.
Using $e^x \le 1+x+x^2$ for $x\le1$, we have $(1+\eta)^k \le e^{\eta k} \le 1+\eta k + (\eta k)^2 \le 1+2\eta k$ because $\eta k < 1$.
Hence, using $\frac{1}{1+x} \ge 1-x$ for all $x > 0$, we have $\frac{1-\delta}{(1+\eta)^k} \ge \frac{1-\delta}{1+2\eta k} \ge (1-\delta)(1-2\eta k) \ge (1-2\delta)$, where the last inequality comes from our choice of $\eta$.

We now argue that each $u\in\tA$ has $x_u \le (1+\eta)M$. Let us define
$v \eqdef \arg\max_{u\in A} x_u$,
the node such that $x_v \le M$ with the greatest value $x_v$. Let $i\in[\ell]$ be the index such that $v\in S_i$. If $i > \tau$, then all $u\in\tA$ verify $x_u \le (1+\eta)^{\tau+1} \le (1+\eta)^i \le x_v \le M$.
On the other hand, if $i \le \tau$, we claim that we must have $i=\tau$. This is a straightforward consequence of the maximality of $x_v$: using that $A \subseteq \bigcup_{j\le i} S_j$, we get
\[ \sum_{j \le i} s_j \ge \sum_{j\le i}|S_j| \ge |A| \ge (1-\delta)|K|\ . \]
So, by minimality of $\tau$, it must be that $\tau \le i$, hence that $\tau=i$.
The lemma follows since each $u\in S_j$ with $j < \tau$ verify $x_u < (1+\eta)^\tau \le x_v \le M$ and each $u\in S_\tau$ verify $x_u \le (1+\eta)^{\tau+1} \le (1+\eta)x_v \le (1+\eta)M$.
\end{proof}

\begin{proof}[Proof of \cref{lem:compute-outliers}]
Applying \cref{lem:filter} to the set $\set{v\in K: \ta_v \le \avganti_K/100}$ (with $\delta=1/100$ and $M=100\avganti_K$), we compute in $O(\log\log n)$ rounds a set $A_1\subseteq K$ of size $(1-3/200)|K|$ such that all $v\in A_1$ verify $\ta_v \le 200\avganti_K$.
By Markov, at most $|K|/100$ nodes verify $e_v > 100\avgext_K$, and at most $|K|/100$ other nodes verify $\errext_v > 100\errext_K$. So, at least $(1-1/50)|K|$ nodes have $\te_v = e_v + \errext_v \le 100(\avgext_K+\errext_K)$.
Applying \cref{lem:filter} to the set $\set{v\in K: \te_v \le 100(\avgext_K + \errext_K)}$ (with $\delta=1/50$ and $M=100(\avganti_K+\errext_K)$), we compute a set $A_2\subseteq K$ of size $(1-3/100)|K|$ such that all $v\in A_2$ verify $e_v \le 200(\avgext_K+\errext_K)$.
Letting $I_K\eqdef A_1\cap A_2$, we have a set for which all nodes verify \cref{eq:inliers} and of size $(1-5/100)|K|$.
\end{proof}

\section{Overcoming Congestion}
\label{sec:congestion}

In this section, we explain how we reduce the bandwidth to $O(\log n)$.
The parts of the algorithm that need to be modified are the sampling of colors when reducing the degree (\cref{lem:sample-log2n}) and the learning of the palettes once the uncolored degree is $O(\log n)$ (\cref{lem:learn-moderatly-dense-log2n,lem:learn-palette-very-dense-log2n}).

\subsection{Color Sampling}
\label{sec:slack-color-logn}

We explain first how nodes sample uniform colors from their palette with $O(\log n)$ bandwidth.
The intuition remains the same as in \cref{lem:sample-log2n}: sampling $\Theta(\log n)$ colors in $\pal{K}$ and dropping the ones used by external neighbors.
The lower bandwidth cost is achieved by using pseudorandom objects (representative hash functions) to sample the colors.
This technique was previously used by \cite{HNT22} to implement the $\deg+1$-list-coloring algorithm of \cite{HKNT22} in \congest.
We extend \cite{HNT22} to comply with our uniformity needs (\cref{eq:uniform-slack-color}).
Intuitively, one can think of a representative hash function as a random function in the sense that for \emph{any sets} $T$ and $P$, the number of collisions under a random function is close to what would be expected from a truly random function.
\Cref{lem:rep-hash-func} gives a formal description. Since the proof is a standard probabilistic argument, we postpone it to \cref{sec:exists-rep-hash-function}.

\begin{restatable}[Representative Hash Functions, Extension of \cite{HNT22}]{lemma}{lemmaRepHashFunc}
\label{lem:rep-hash-func}
Let $\alpha, \beta\in(0,1/8)$ and $\lambda \ge 1$ an integer be such that $\alpha \le \beta$ and $\lambda \ge \Omega(\alpha^{-1}\beta^{-2}\log n)$. Let $\calU$ be a finite set.
There exists a family $\hashrep$ of $F=\Theta(\beta\lambda^2\log|\calU|\poly(n))$ functions $\calU\to[\lambda]$ and a $\sigma \le \lambda$ verifying $\sigma \ge \Theta(\beta^{-2}\alpha^{-1}\log n)$ such that
\begin{enumerate}
    \item\label{part:collisions} For any pair $(x,y)\in \calU\times[\lambda]$ and \emph{disjoint} sets $T,P\subseteq \mathcal{U}$ with sizes $|T|,|P|\le \beta\lambda$ and $|T|\ge \alpha\lambda$,
    when sampling a $h\in\hashrep$ such that $h(x)=y$, with high probability, \cref{eq:hash-func} holds:
    \begin{equation}
    \label{eq:hash-func}
        \card{(h(T)\setminus h(P))\cap [\sigma]} \ge \frac{\sigma|T|}{\lambda}(1-8\beta)\ .
    \end{equation}
    \item For all pairs $(x,y)\in \calU\times[\lambda]$ there are $\frac{1\pm 1/2}{\lambda}F$ functions in $h\in \calH$ such that $h(x)=y$.
\end{enumerate}
\end{restatable}

\begin{lemma}
\label{lem:uniform-hash-color}
\label{lem:learn-moderately-logn}
For all nodes $v\in K \in\Kmod$,
there is a randomized algorithm (using $O(\log n)$ bandwidth) sampling a random color $\rC_v\in\pal{v}\cup\set{\bot}$ (where $\bot$ represents failure) such that $\Pr(\rC_v = \bot)\le 1/\poly(n)$ and for all $c\in \pal{v}$ we have $\Pr(\rC_v = c)\le \frac{\kappa}{\hatd(v) + \Omega(\avganti_K+\avgext_K+\te_v)}$ for some universal constant $\kappa > 0$.
\end{lemma}

\begin{proof}
We begin by explaining the algorithm and implementation, and we later explain why it samples colors with the right probability.
Let $\hashrep$ be a family of representative hash functions for parameters $\lambda,\alpha, \beta$ to be defined later.
Node $v$ samples a random $h\in\hashrep$ and broadcast it to its neighbors, which can be done in $O(1)$ rounds because $\lambda, |\hashrep|\le \poly(n)$.
Similarly to \cref{lem:free-color}, we can split $K$ randomly into groups $T_1, \ldots, T_k$ with $k = \Theta(|K|/\log n)$. Nodes of $T_i$ learn about a range of $\Theta(\log n)$ colors from the color space.
Since $v$ has a distance-2 neighbor in each $T_i$ (\cref{lem:random-groups}), each distance-1 neighbors $u$ knows about a set $S_u\subseteq \pal{K}$ such that $\bigcup_{u\in N(v)} S_u=\pal{K}$.
Each $u\in N(v)$ send $\set{i\in[\sigma]: i\in h(S_u)}$ in $O(1)$ rounds using a bitmap (because $\sigma=\Theta(\log n)$.
Using a bitwise OR, node $v$ computes $A_v = \set{i\in[\sigma]: i\in h(\pal{K})}$.
Similarly, node $v$ can learn $B_v = \set{i\in[\sigma]: i\in h(\col(\Ntwo(v)\setminus K))}$ the hash values of external neighbors.
Then $v$ samples an index $i_v\in A_v \setminus B_v$ uniformly at random.
% By \cref{lem:free-color}, each $v$ learns $c_v$ the $i_v$-th color of $\pal{K}$.
At least one of its neighbors can tell $v$ which color in $\pal{K}$ hashes to $i_v$ (breaking ties arbitrarily).
Note that this color must be in their palette.
Let the random color $\rC_v$ be that color, and $\rC_v=\bot$ if $A_v\setminus B_v=\emptyset$.
\

Letting $\lambda\eqdef (100/\gamma_1)|\pal{K}|$, $\beta \eqdef 1/100$ and $\alpha \eqdef \frac{\gamma_1}{100}\cdot\gamma_2$ where $\gamma_1, \gamma_2\in(0,1)$ are the constants from \cref{lem:slack} such that 1) $|\pal{K}\cap\pal{v}| \ge \gamma_1e_v$ and 2) $|\pal{K}\cap\pal{v}| \ge \gamma_2|\pal{K}|$. (Recall nodes can compute $|\pal{K}|$ in $O(1)$ rounds, \cref{lem:free-color}.)
Let $T=\pal{K}\cap\pal{v}$ and $P=\col(\Ntwo(v)\setminus K)$.
Clearly, $T\cap P=\emptyset$ since $T\subseteq \pal{v}$ and $P\cap\pal{v}=\emptyset$, both by definition.
By our choice of parameters, we have $|T| \le |\pal{K}| \le \beta\lambda$, and $|T| \ge \alpha\lambda = \gamma_2|\pal{K}|$.
On the other hand, $|P| \le e_v \le \beta\lambda$.
Hence, by definition of $\hashrep$, with high probability,
\begin{equation}
\label{eq:size-set}
 |A_v\setminus B_v| = |(h(T)\setminus h(P))\cap[\sigma]|\ge \frac{\sigma|T|}{\lambda}(1-8\beta) \ge \alpha(1-8\beta)\sigma = \Theta(\log n)\ .
\end{equation}
\cref{eq:size-set} implies $\rC_v\neq\bot$ with high probability.
We turn to the uniformity. Fix some $c\in\pal{v}\cap\pal{K}$ (if $c\notin\pal{K}$ then $\Pr(\rC_v=c)=0$), then
\begin{align*}
\Pr(\rC_v = c) &= \sum_{i\in [\sigma]}\Pr(\text{pick index $i$ in $A_v\setminus B_v$}~|~h_v(c)=i)\times\Pr(h_v(c)=i) \\
&\le \frac{1}{\alpha(1-8\beta)\sigma} \times \frac{2\sigma}{\lambda} \\
&= \frac{1}{4\alpha(1-8\beta)|\pal{K}|} \tag{by definition of $\lambda$} \\
&\le \frac{\kappa}{\hatd(v)+\Omega(\avganti_K+\avgext_K+\te_v)}\qquad\qquad \text{for}\quad \kappa=\frac{1}{4\alpha(1-8\beta)} \tag{by \cref{lem:slack}}\ ,
\end{align*}
where the first inequality holds because \cref{eq:size-set} is true for any conditioning on one pair $h(c)=i$, and there are at most $2F/\lambda$ functions $h\in \hashrep$ such that $h(c)=i$ for any pair.
\end{proof}

\subsection{Learning Palette}
\label{sec:learn-palette-logn}

In this section, we reduce the bandwidth of the last step of our algorithm (\cref{sec:learn-palette-log2n}).
Our algorithm follows the same overall structure: nodes need to learn $\poly\log n$ colors from the clique palette, colors used by their external neighbors and (in very dense almost-cliques) the colors of their anti-neighbors.

\paragraph{Moderately Dense Almost-Cliques.}
When \cref{lem:slack} applies, we use the same approach as \cref{lem:learn-moderatly-dense-log2n} but with representative hash functions. The main difference with \cref{lem:uniform-hash-color} is that nodes need to infer $\Theta(\log n)$ colors from hashes, instead of one.

\begin{lemma}
\label{lem:moderately-dense-logn}
Let $v\in K\in \Kmod$, and $\hatd(v)< C'\log n$ for some large enough constant $C' > 0$.
There is a $O(1)$-round algorithm (using $O(\log n)$ bandwidth) for all such $v$ to learn at list $L(v)\subseteq \pal{v}$ containing $\hatd(v)+1$ colors.
\end{lemma}

\begin{proof}
Notice that we are considering the sames nodes as in \cref{lem:uniform-hash-color}.
Therefore, we implement the same algorithm to sample indices in $[\sigma]$.
In particular, \cref{eq:size-set} still holds, and for a large enough constant in $\sigma$, we have $|A_v\setminus B_v| \ge \alpha(1-8\beta)\sigma \ge \Theta(\log n) \ge \hatd(v)+1$.
% Since nodes know the clique palette at this point (\cref{fact:small-palette}), they know which colors correspond to which index and can compute the list $L(v)$ from $A_v\setminus B_v$.
To learn which color corresponds to which selected hash, we pick nodes $u_{v,i}$ in $K$, one of each pair $(v,i)$ with $v\in \uncolored{K}$ and $i\in A_v\setminus B_v$ selected by $v$. We use many-to-all broadcast (\cref{lem:many-to-all}) for 1) each $v\in\uncolored{K}$ to broadcast $h_v$, and 2) for each $u_{v,i}$ to broadcast message $(v,i,c)$ where $c\in\pal{K}$ such that $h_v(c)=i$. Since $|\uncolored{K}|=O(\log n)$ and each $v \in \uncolored{K}$ selects $O(\log n)$ indices, there are $O(\log^2 n)$ such pairs, which allows many-to-all broadcast to work when $\Delta \ge \Omega(\log^{3.5} n)$. Also note that $u_{v,i}$ learns such a $c\in\pal{K}$ the same way nodes did in \cref{lem:uniform-hash-color}.
\end{proof}

\paragraph{Very Dense Almost-Cliques.}
For very dense almost-clique, the $O(\log^2 n)$ bandwidth was only needed when learning which colors used by external neighbors were in the clique palette.
We reduce the bandwidth by reducing exponentially the number of bits needed to describe colors in $\pal{K}$.
We use pairwise independent hash functions.

\begin{definition}[Almost Pairwise Independent Hash-Functions {\cite[Problem 3.4]{Vadhan12}}]
A family of functions $\hashpwi$ mapping elements from $[N]$ to $[M]$ is $\delta$-almost pairwise independent if for every $x_1\neq x_2\in [N]$ and $y_1,y_2\in [M]$, we have
\[ \Pr(h(x_1)=y_1 \wedge h(x_2)=y_2) \le \frac{1+\delta}{M^2} \ . \]
There exists a family $\hashpwi$ of $\delta$-almost pairwise independent from $[N]$ to $[M]$ such that choosing a random function in $\hashpwi$ requires $O(\log\delta^{-1} + \log\log N + \log M)$ random bits.
\end{definition}

\begin{lemma}
\label{lem:learn-palette-very-dense-logn}
When $v\in K\in\Kvery$, there is a $O(\log\log n)$-round randomized algorithm (using $O(\log n)$ bandwidth) for all uncolored nodes to learn $\hatd(v)+1$ colors from their palettes.
\end{lemma}

\begin{proof}
We run first the algorithm of \cref{lem:learn-colors-clique-palette} to learn the clique-palette, which uses only $O(\log n)$ bandwidth. We describe the algorithm when nodes learn $\pal{K}$. The case where nodes learn a set $D \subset \pal{K}$ of $\Theta(\log^2 n)$ colors works the same with $L'(v)=D$ for all $v\in\uncolored{K}$ (ignoring colors of anti-neighbors).

Following the first steps of \cref{lem:learn-palette-very-dense-log2n}, an uncolored node knows the set $L'(v)\eqdef\pal{K}\cup\col(K\setminus \Ntwo(v))$ of $O(\log^2 n)$ colors constituted of $\pal{K}$ as well as of the $O(\log n)$ colors of its anti-neighbors.
% Let $h_v\in\hashpwi$ mapping $[\Delta^2+1]$ to $[\Theta(\log^4 n)]$ such that it has no collisions on $L'(v)$.
Let $\hashpwi$ be a family of $\delta$-almost pairwise independent hash functions from $[\Delta^2+1]$ to $[M]$, with $\delta$ an arbitrary small constant and $M$ to be defined later.
By union bound, the probability that a random function in $\hashpwi$ has a collision on one pair $c,c'\in L'(v)$ with $c\neq c'$ is at most $\frac{(1+\delta)|L'(v)|^2}{M^2} < 1$ where the last inequality holds if $M=\Theta(|L'(v)|)=\Theta(\log^2 n)$.
Hence, for these choice of parameters, there exists a $h_v\in \hashpwi$ without collisions on the colors of $L'(v)$.
We update lists to $L(v)\eqdef \set{c\in L'(v): h_v(c)\notin h_v(\col(\Ntwo(v)\setminus K))}$ by removing colors hashing to the same value as one of the external neighbors.
If an external neighbor adopted a color $c\in L'(v)$, clearly $c\notin L(v)$, i.e., $L(v)\subseteq \pal{v}$.
Since each external neighbor removes at most one color from $L'(v)$ (because $h$ is collision free on $L(v)$), the set $L(v)$ has size at least $|\pal{K}| + \colored{a}_v - \colored{e}_v \ge \hatd(v) + 1$ (for the same reason as in \cref{lem:learn-palette-very-dense-log2n}).
This algorithm uses only $O(\log n)$ bandwidth because describing the hash function $h_v$ requires $O(\log\log n)$ bits and sending $e_v=O(\log n)$ hash values on an edge requires $O(\log n)\times O(\log \log n)$ bits, hence, $O(\log\log n)$ rounds.
\end{proof}

\bibliographystyle{alpha}
\bibliography{references}

\appendix

\section{Concentration Inequality}

We use the following variants of Chernoff bounds for dependent random variables. The first one is obtained, e.g., as a corollary of Lemma 1.8.7 and Theorems 1.10.1 and 1.10.5 in \cite{Doerr2020}.

\begin{lemma}[Martingales]\label{lem:chernoff}
Let $\{X_i\}_{i=1}^r$ be binary random variables, and $X=\sum_i X_i$.
Suppose that for all $i\in [r]$ and $(x_1,\ldots,x_{i-1})\in \{0,1\}^{i-1}$ with $\Pr\parens{X_1=x_1,\dots,X_r=x_{i-1}}>0$, $\Pr\parens{X_i=1\mid X_1=x_1,\dots,X_{i-1}=x_{i-1}}\le q_i\le 1$, then for any $\delta>0$,
\begin{equation}\label{eq:chernoffless}
\Pr\parens*{X\ge(1+\delta)\sum_{i=1}^r q_i}
\le \exp\parens*{-\frac{\min(\delta,\delta^2)}{3}\sum_{i=1}^r q_i}\ .
\end{equation}
Suppose instead that $\Pr\parens*{X_i=1\mid X_1=x_1,\dots,X_{i-1}=x_{i-1}}\ge q_i$, $q_i\in (0,1)$ holds for $i, x_1, \ldots, x_{i-1}$ over the same ranges, then for any $\delta\in [0,1]$,
\begin{equation}\label{eq:chernoffmore}
    \Pr\parens*{X\le(1-\delta)\sum_{i=1}^r q_i}\le \exp\left(-\frac{\delta^2}{2}\sum_{i=1}^r q_i\right)\ .
\end{equation}
\end{lemma}

\section{Pseudo-Code}
\label{sec:pseudo-code}

In this section, we give pseudo-code for some of the basic routines of distributed coloring.

\begin{algorithm}[H]\label{alg:trycolor}
\caption{$\trycolor$(a node $v$, a color $c_v$)}
\DontPrintSemicolon
Send $c_v$ to $N(v)$.

If there is a $u\in \Ntwo(v)$ which has $\ID(u) < \ID(v)$ and $c_u=c_v$ then there is a relay $r\in N(v)\cap N(v)$ which knows about it. It informs $v$ and $u$ about their failure.

\textbf{if} $v$ did not fail, then it adopts the color $c_v$, i.e. define $\col(v)=c_v$.
\end{algorithm}

\begin{algorithm}\label{alg:slackgeneration}
\caption{\slackgeneration}
\DontPrintSemicolon
Each node chooses to be active independently with probability $\pg=1/20$\;
Active nodes sample $c_v\in[\Delta^2+1]$ uniformly at random and call $\trycolor(c_v)$ \;
\end{algorithm}

\section{Colorful Matching}
\label{sec:colorful-matching}

\begin{algorithm}[ht]
\caption{\matching (for a fixed clique $K$).}
\label{alg:matching}
\Input{a clique $K$ and a constant $0 < \beta < 1/(18\epsilon)$.}
\Output{a colorful matching $M$ of size $\beta\cdot\avganti_K$.}

\Repeat{$O(\beta)$ times}{
Sample each color with probability $p=1/(4\Delta)$. If a node sampled more than one  color, it drops all its colors. If the sampled color $c$ is used by a neighbor or by an edge of $M$, the node drops the color.

For each color we did not drop at the previous step, nodes keep their color if there is at least one anti-neighbor with this color. If the same color is used by two (or more) anti-edges, we keep the one with the smallest \ID.
}
\end{algorithm}

\cref{alg:matching} computes a colorful matching of size $\beta\avganti_K$ with probability $1-e^{\Theta(\avganti_K)}$.
We refer the reader to \cite{FGHKN22} for a proof of correctness.
A node can easily check if the color is already used by a node in the colorful matching (\cref{lem:free-color}).
We argue nodes can check if an anti-neighbor sampled the same color as them in $O(1)$ rounds.

We split the color space in ranges $R_1, \ldots, R_k$ with $k=\Theta(\Delta^2/\log n)$.
Observe that the set $T_i$ of nodes sampling a color in $R_i$ is a random group; hence, has size $\Theta(\log n)$ and 2-hop connects $K$.
Compute a BFS tree in each $T_i$.
Using bitmaps, each $u\in T_i$ learns which colors in $R_i$ are sampled in each tree of $T_i\setminus{u}$.
This is because relays of $u\in T_i$ know which color was sampled by $u$, so they can filter out its contribution to the bitmap (if needed).
Since anti-neighbors that sample the same color are in the same group, all nodes learn if there exists an anti-neighbor with the same color as theirs.
If more than one anti-edge has color $c$, using aggregation on the BFS tree in $T_i$, we can select the one with the smallest pair of \IDs.

Using the prefix-sum algorithm and random groups (\cref{lem:prefix-sum,lem:free-color}), an edge in $M$ can learn how many edges $M$ contains with a smaller color. In particular, they can adjust the size of $M$ to be exactly $\beta\avganti_K$ (in a case it contains more edges than needed).

\section{Almost-Clique Decomposition}
\label{sec:acd}

Almost-clique decomposition are commonly computed by classifying edges as \emph{friendly} (or \emph{un}friendly) if neighbors share a $(1-\epsilon)$ fraction of their neighborhoods.
Computing $\epsilon$-friendly edges exactly is too expensive; \cite{ACK19} showed it was enough to compute a weaker predicate distinguishing $\epsilon$-friends from non-$(c\cdot\epsilon)$-friends where $c > 1$ is some constant.
This predicate is called $\epsilon$-buddy and is the one implemented by \cite{HNT22} in \congest. A node with many buddies is called \emph{popular}, and the popular predicate is what classifies nodes as either sparse or dense.
We refer the reader to \cite[Appendix B in full version]{HKMN20} for more details on $\epsilon$-buddies.
\newcommand{\estimatesim}{\alg{EstimateSimilarity}}
One of the main technical contributions of \cite{HNT22} is a distributed algorithm \estimatesim:

\begin{lemma}[{Estimate Similarity, \cite[Lemma 2]{HNT22}}]
\label{lem:estimate-similarity}
Fix any $\epsilon > 0$. Suppose nodes $u$ and $v$ know respectively sets $S_u$ and $S_v$ from a universe $\calU$.
The randomized algorithm \estimatesim uses $O(1)$ rounds with messages of size $O(\log n/\epsilon^4 + \log\log|\calU| + \log\max(|S_u|,|S_v|))$ such that w.h.p.\ $u$ and $v$ know $|S_v\cap S_u|$ up to an error $\epsilon\max(|S_u|,|S_v|)$.
\end{lemma}

At distance-1, by letting $S_u=N(u)$ and $S_v=N(v)$, this algorithm computes the $\epsilon$-buddy predicate, thereby the almost-clique decomposition.
To compute the estimates, each node $v$ samples a hash function $h_v$ which can be described with $O(\log n)$ bits.
It then compute $T_u$, the set of values $x\in S_u$ that 1) hash to a value $h(x) \le \Theta(\log n)$ and 2) collide with no other values in $S_u\setminus{x}$.
It then sends $h(T_u)$ to its neighbors.

For computing almost-clique decompositions, it was then remarked in \cite{FGHKN23} that this algorithm could be simplified by letting each node pick a random value in $[O(\Delta/\eps^4)]$ and using this value as ``hash'' instead of each node picking a hash function. As this process of letting each node pick a random value $h(v)$ is formally equivalent to picking a global random hash function $h:V \to [O(\Delta/\eps^4)]$, the analysis of \cref{lem:estimate-similarity} works immediately using this function instead of individually picked functions $h_u,h_v$. In \cite{FGHKN23}, this allowed for the computations of almost-clique decompositions with only broadcast communication.

The same observation essentially also allows for computing almost-clique decompositions at distance-2.
The only subtlety at distance-2 is that, since $v$ does not know it neighborhood, it can mistake two nodes hashing to the same value with having two paths to the same node.
This issue is solved by being very conservative and ignoring hashes received multiple times.
Note that, by doing so, we underestimate $|S_v| = d(v)$ by $\td(v)-d(v)$ because each time some node $u\in \Ntwo(v)$ with two distinct 2-paths to $v$ hashes to a value below $O(\log n)$, two relays send this value to $v$, making $v$ remove it from $T_v$.
A node with $\td(v)-d(v) \ge \epsilon\Delta^2$ will underestimate its number of friends, note however that such a node is sparse. Underestimating the number of friends of a sparse node does not cause any harm as even without such errors, it does not have enough friends to belong to an almost-clique.
On the other hand, a node with $\td(v) - d(v) \le \epsilon\Delta^2$ will lose only $\epsilon\Delta^2$ when estimating its number of friends.
Fortunately, this much error can be tolerated when we implement the $\epsilon$-buddy predicate.
The algorithm is summarized in \cref{alg:estimate-sim}.

\begin{algorithm}[ht]
\caption{\alg{Buddy} and \alg{Popular}}
\label{alg:estimate-sim}
Let $\calH$ be a family of representative hash functions $\calU\to[\lambda]$ with parameters
\[ \lambda \eqdef  8\Delta^2/\epsilon, \quad \beta\eqdef \epsilon/4, \quad \alpha \eqdef \epsilon^2/8 \ , \quad \sigma = \Theta(\eps^{-4}\log n) \ . \]

Each $v$ samples a random value $h(v) \in [\lambda]$ and sends it to its neighbors.

Each relay $r\in N(v)$ of $v$ computes $h(N_G(r))\cap[\sigma]$

Let $T_v = [\sigma] \cap \bigcup_{r\in N_G(u)} \parens*{ h(N_G(r)) \setminus h(\bigcup_{r'\neq r} N_G(r')) }$

Each $v$ broadcasts $T_v$.

A shared relay $r\in N(v)\cap N(u)$ computes $|T_u \cap T_v|\lambda/\sigma$ and declare $v$ and $u$ friends if it is at least $(1-\Theta(\epsilon))\Delta^2$.

Each $r\in N_G(v)$ sends to $v$ how many friends it detected in $N_G(r)$.

If $v$ received notice of at least $(1-\Theta(\epsilon))\Delta^2$ friends from its relays, it declares itself popular.
\end{algorithm}

\section{Proof of
\texorpdfstring{\cref{lem:rep-hash-func}}
{Lemma~\ref{lem:rep-hash-func}}}
\label{sec:exists-rep-hash-function}

\lemmaRepHashFunc*

\paragraph{Quick Review of HNT.}
We begin by reviewing the work of \cite{HNT22} (henceforth referred to as HNT) on which this part is heavily based.
For some function $h:\calU\to [\lambda]$ and sets $A,B\subseteq \calU$, they introduce sets
\begin{itemize}
\item $A|_h^{\le\sigma}\eqdef h^{-1}([\sigma])\cap A$ the values to hash $\le\sigma$ through $h$;
\item $A\wedge_h^{\le\sigma} B \eqdef \set{x\in A: h(x)\in [\sigma]\cap h(B\setminus\set{x})}$ the values $x\in A$ which collide with some $B\setminus\set{x}$ through $h$; and
\item $A\neg_h^{\le\sigma} B\eqdef (A|_h^{\le\sigma})\setminus (A\wedge_h^{\le\sigma} B)$ the values $x\in A$ to hash $\le\sigma$ without colliding with some element in $B\setminus\set{x}$.
\end{itemize}
Henceforth, we will only consider sets $A\subseteq B$, such that $|h(A\neg_h^{\le \sigma} B)|=|A\neg_h^{\sigma} B|$ (because each index in $[\sigma]$ has a unique pre-image in $A$).
Note that $|A\neg_h^{\le\sigma} B| \ge |A_h^{\le\sigma}| - |A\wedge_h^{\le\sigma} B|$.
They key lemma from HNT is the following: (we will use a $\beta'\neq\beta$ from \cref{lem:rep-hash-func} to be defined precisely later)

\begin{lemma}[Claim 1 from HNT]
\label{lem:hnt}
Let $\alpha,\beta',\nu\in (0,1)$ and $\lambda$ such that $
\alpha \le \beta'$ and $\lambda \ge \Omega(\alpha^{-1}\beta'^{-2}\log(1/\nu))$.
There exists a $\sigma = \Theta(\beta'^{-2}\alpha^{-1}\log(1/\nu)) \le \lambda$ such that
if $A,B\subseteq \calU$ verify $|A|,|B|\le\beta\lambda$ and $|A|\ge\alpha\lambda$, then
a random function $h:\calU\to[\lambda]$ verify that $|A\neg_h^{\le\sigma} B| \ge \frac{\sigma|A|}{\lambda}(1-3\beta')$ with probability $1-\nu/2$.
\end{lemma}

\paragraph{Proof of \cref{lem:rep-hash-func}.}
\cref{lem:hnt} states that a random function essentially verify our requirements.
We now give the existential argument for the existence of the family $\hashrep$ described in \cref{lem:rep-hash-func}.

Sample $F$ functions $h_1, h_2, \ldots, h_F:[\lambda]\to \calU$ uniformly at random.
Let $\calH=\set{h_1, \ldots, h_F}$.
For each pair $(x,y)\in \calU\times [\lambda]$, we define the set $\calH(x,y)\eqdef\set{h\in\calH\text{ such that } h(x)=y}$.
The following fact is a straightforward consequence of the uniformity of functions and Chernoff bound.

\begin{fact}
\label{fact:prop1}
For any pair $(x,y)\in [\lambda]\times \calU$, w.p. $1-e^{-\Theta(F/\lambda)}$, we have $|\calH(x,y)| \in \frac{1\pm 1/2}{\lambda}F$.
\end{fact}

Fix a pair $(x,y)\in [\lambda]\times\calU$ and condition on an arbitrary set $\calH(x,y)$ of size at least $F/(2\lambda)$.
Let $T, P\subseteq \calU$ be any set such that $|T| \ge \alpha\lambda$ and $|T|,|P|\le \beta\lambda$ (such as required by \cref{lem:rep-hash-func}).
We define $A\eqdef T$ and $B\eqdef T\cup P$, such that $|A|\ge \alpha\lambda$ and $|B| \le 2\beta\lambda \le \beta'\lambda'$ where $\beta'\eqdef2\beta$.
The restriction $h_{|\calU\setminus\set{x}}$ of $h$ is a uniform function $\calU\setminus\set{x}\to[\lambda]$; hence, \cref{lem:hnt} applies.
When adding the fixed value $h(x)=y$, we remove at most one hash value. Hence, we have the bound $|h(T)\setminus h(P)\cap[\sigma]|\ge|A\neg_h^{\sigma} B| \ge \frac{\sigma|A|}{\lambda}(1-3\beta')-1\ge \frac{\sigma|A|}{\lambda}(1-4\beta')=\frac{\sigma|A|}{\lambda}(1-8\beta)$ for $\sigma=\Theta(\beta'^{-2}\alpha^{-1}\log 1/\nu)$ large enough.
Since a function is bad for sets $A,B$ with probability at most $\nu/2$ and $|\calH(x,y)| \ge F/(2\lambda)$, the Chernoff bound gives the following:

\begin{fact}
\label{fact:prop2}
The number of functions in $\calH(x,y)$ to verify \cref{eq:hash-func} is $(1-\nu)|\calH(x,y)|$ with probability $1-e^{-\Theta(\nu F/\lambda)}$.
\end{fact}

Define $\evt_1(x,y)$ the event that \cref{fact:prop1} fails and $\evt_2(x,y)$ the event that \cref{fact:prop2} fails.
We show that with probability strictly less than 1 that no $\evt_1(x,y)$ nor $\evt_2(x,y)$ occurs.
A simple union bound gives
\begin{align*}
    \Pr(\exists x,y:~\evt_1(x,y)~\vee~\evt_2(x,y))
    &\le \sum_{x,y} 2\Pr(\evt_1(x,y)) + \Pr(\evt_2(x,y)~|~\overline{\evt_1}(x,y)) \\
    &\le \lambda|\calU| \parens*{2e^{-\Theta(F/\lambda)} + |\calU|^{\Theta(\beta\lambda)} e^{-\Theta(\nu F/\lambda)}} < 1
\end{align*}
for some large enough $F=\Theta(\lambda^2\nu^{-1}\log|\calU|)$.
By the probabilistic method, there exists a family $\hashrep$ such as described in \cref{lem:rep-hash-func}.

\section{Low Degree Algorithm}
\label{sec:appendix-small-degree}

We explain in this section how results in
\cite{HKMN20,MPU23,GK21} are combined to obtain an $O(\log^5 \log n)$ algorithm for $\poly(\log n)$-sized instances of deg+1-list-coloring on square graphs, as stated in \cref{lem:small-degree}.

\lowDegColoring*

The algorithm of \cite{HKMN20} first reduces to two instances of maximum degree $O(\log n)$, similar to \cite{BEPS}, using $O(\log \log n)$ color tries. As the instances to solve are connected in $G^2$ rather than $G$, they then add some relay nodes (which they call Steiner nodes) to each instance to make them connected in $G$. These relays are few enough that they do not increase much the size of the instances. A distance-2 network decomposition of this graph is then computed, which is where recent results on network decompositions first improve the runtime of \cite{HKMN20}.

The distance-2 network decomposition can be computed in $\widetilde{O}(\log^3 \log n)$ \congest rounds, using a recent result of \cite{MPU23} adapting the network decomposition of \cite{GGHIR23} to power graphs with a mild dependency on the ID space. At this point, we can iterate through the colors of the network decomposition, and color each cluster of it following the same method as in \cite{HKMN20}.
In each cluster, nodes reduce their color space by choosing a hash function from $[\Delta^2+1]$ to $\poly(\log n)$ using a standard derandomization technique (method of conditional expectation). Each cluster then colors itself using a deterministic algorithm for degree+1-list-coloring. Recent improvements for this problem allow this step to run in $O(\log^4 \log n)$ within each cluster, where one $O(\log \log n)$ factor comes from simulating the algorithm of~\cite{GK21} in the distance-2 setting. More precisely, the algorithm only needs $O(\log^3 \log n)$ communication rounds on $G^2$ with $O(\log \log n)$ bandwidth, which can be implemented in $O(\log^4 \log n)$ with $G$ as communication graph, with $\Theta(\log n)$ bandwidth and using that the degree is at most $O(\log n)$.

Putting everything together, the total cost of the algorithm is $O(\log^5 \log n)$. The complexity from computing the network decomposition~\cite{GGHIR23,MPU23} is dominated by that of coloring the decomposition's clusters afterwards, decomposed as follows: one $O(\log \log n)$ from iterating through the colors of the network decomposition, $O(\log^3 \log n)$ to run \cite{GK21} in each cluster, with an $O(\log \log n)$ overhead to run the algorithm with communication graph $G$ instead of~$G^2$.

%[{\cite[Lemma 3.12+3.15]{HKMN20}} + {\cite[Appendix A]{MPU23}} + \cite{GK21}]

\end{document}